\documentclass[12pt]{article}

\usepackage[T1]{fontenc}
\usepackage[utf8]{inputenc}
\usepackage{textcomp}
\usepackage{lmodern}

\usepackage{amsmath,amssymb,amsthm}

\usepackage{algorithm}
\usepackage{algpseudocode}

\newtheorem{theorem}{Theorem}[section]
\newtheorem{lemma}[theorem]{Lemma}

\newtheorem{remark}[theorem]{Remark}

\usepackage{graphicx}
\usepackage{booktabs}
\usepackage{array}
\usepackage{longtable}
\usepackage{caption}
\usepackage{subcaption}
\usepackage{float}
\usepackage{xcolor}

\usepackage[numbers]{natbib}
\usepackage{url}
\usepackage[colorlinks=true,
linkcolor=blue,
citecolor=blue,
urlcolor=blue]{hyperref}

\usepackage{authblk}

\newcommand{\anon}{1}

\begin{document}
	
	\def\spacingset#1{\renewcommand{\baselinestretch}{#1}\small\normalsize}
	\spacingset{1}
	
\if1\anon

\title{\bf Penalized KLIC Model Selection for the Generalized Method of Moments in Longitudinal Data with Time-Dependent Covariates}

\author[1]{Mahmud Hasan*}
\author[2]{Mathias Nthiani Muia}
\author[3]{Mous-Abou Hamadou}
\author[1]{Niloofar Ramezani}

\affil[1]{Department of Biostatistics, Virginia Commonwealth University, USA}
\affil[2]{Department of Mathematics and Statistics, University of South Alabama, USA}
\affil[3]{Department of Mathematics and Statistics, Quinnipiac University, USA}

\affil[ ]{\footnotesize
	\texttt{hasanm10@vcu.edu} \quad
	\texttt{mnmuia@southalabama.edu} \quad
	\texttt{mhamadou@quinnipiac.edu} \quad
	\texttt{ramezanin2@vcu.edu}
}

\date{}
\maketitle

\fi
	
	\if0\anon
	\begin{center}
		{\LARGE\bf Penalized KLIC Model Selection for the Generalized Method of Moments in Longitudinal Data with Time-Dependent Covariates}
	\end{center}
	\medskip
	\fi
	
	\begin{abstract}
		Model selection plays an important role in longitudinal data analysis, especially when models are estimated using the generalized method of moments (GMM) in the presence of time-dependent covariates. In this setting, the number of valid moment conditions can grow quickly and may lead to over-parameterized models. The Kullback--Leibler Information Criterion (KLIC) has been proposed as a model-selection tool for this framework; however, the original KLIC criterion may favor overly complex models when the number of parameters or valid moment conditions increases. To address this limitation, this study proposes two penalized versions of KLIC that incorporate penalties based on both the number of model parameters and the number of valid moment conditions. The proposed criteria are referred to as the Moment--Parameter Product Penalty KLIC (MPPP--KLIC) and the Logarithmic Penalty KLIC (LP--KLIC). These criteria provide a theoretically motivated mechanism for balancing model fit and model complexity in GMM-based longitudinal models. Through an extensive simulation study involving both binary and continuous response settings, the proposed criteria are shown to improve the ability of KLIC to distinguish among competing models and to reduce the selection of over-parameterized models. The performance of the proposed methods is further illustrated using the Filipino Child Morbidity dataset, a longitudinal study of child health in the Philippines. The results show that the proposed penalized criteria provide stable and interpretable model rankings and consistently identify age as the most important predictor of child morbidity. Overall, the proposed penalized KLIC criteria offer practical and theoretically grounded tools for model selection in GMM-based longitudinal data analysis with time-dependent covariates.
	\end{abstract}
	

 \section{Introduction}

Longitudinal data arise when repeated measurements are collected on the same subjects over time using the same or comparable instruments. Such data are common in public health, education, economics, and the social sciences because they allow researchers to examine temporal change and within-subject variability. Well-known examples include the National Health and Nutrition Examination Survey (NHANES), which has followed a panel of participants annually since 1999 to record repeated health and nutritional measurements, and the National Education Longitudinal Study (NELS), initiated in 1988, which tracked a cohort of students across multiple survey waves to evaluate educational outcomes and development. Unlike cross-sectional data, longitudinal observations violate the assumption of independent and identically distributed observations because repeated measurements from the same subject are inherently correlated. As a result, statistical methods used for longitudinal data must explicitly account for within-subject correlation in order to obtain valid inference. Two major modeling frameworks are commonly used in this context: marginal models and conditional models. Marginal models provide population-averaged interpretations by modeling the mean response and the within-subject correlation separately \citep{Zeger1986,Zeger1992,Diggle2002}, whereas conditional or subject-specific models explain correlation through individual-level heterogeneity using random effects \citep{Zeger1988,Fitzmaurice2011,Hedeker2006}. 

 In marginal models, the expected response at time $t$ is modeled as a function of covariates observed at the same time point. Marginal models are preferred when population-average interpretations are desired or when predictions rely only on current covariates \citep{Diggle2002,Pepe1994}. One of the most widely used estimation methods for marginal models with correlated responses is the generalized estimating equations (GEE) approach introduced by \citep{Liang1986}. The GEE method uses estimating equations based on a working correlation structure and requires only specification of the first two moments of the response rather than the full likelihood. Although GEE estimators are consistent even when the working correlation structure is misspecified, efficiency may be substantially reduced when time-dependent covariates are present and the independence working correlation is imposed \citep{Pepe1994,Fitzmaurice1995,Diggle2002}.

A common feature of longitudinal studies is the presence of time-dependent covariates (TDCs), whose values change over repeated measurements. Examples include age, body mass index, treatment exposure, and other evolving subject characteristics. The temporal behavior of these covariates has important implications for estimation and inference. Prior research has shown that when time-dependent covariates are present, the use of GEE with non-independent correlation structures may lead to inconsistent estimators unless strong assumptions are satisfied \citep{Pepe1994,Fitzmaurice1995}. To address this issue, \citep{Lai2007} classified time-dependent covariates into several types according to the feedback structure between the response and covariate processes. In particular, Type~I, Type~II, and Type~III covariates generate different sets of valid estimating equations depending on the temporal relationships between covariates and residuals. When covariates are of Type~I or Type~II, additional valid moment conditions exist that are not exploited by independence GEE. 

The Generalized Method of Moments (GMM), originally developed by \citep{Hansen1982}, provides a natural framework for incorporating these additional valid moment conditions. Instead of relying on likelihood or quasi-likelihood functions, GMM estimation is based on moment conditions whose expectations equal zero. By utilizing all valid moment conditions implied by the temporal structure of the covariates, GMM can yield more efficient parameter estimates than GEE when time-dependent covariates of Type~I or Type~II are present, while remaining consistent for Type~III covariates \citep{Lai2007}. As a result, GMM has become an attractive alternative for estimating marginal models in longitudinal studies with complex covariate structures.

Despite these advantages, model selection for GMM-based longitudinal models remains relatively underdeveloped. In likelihood-based modeling, several widely used goodness-of-fit and model-selection criteria exist, including the model deviance \citep{Agresti1990}, Akaike's Information Criterion (AIC) \citep{Akaike1973,Akaike1974}, the corrected AIC (AICc) for small samples \citep{Sugiura1978,Hurvich1989}, and the Bayesian Information Criterion (BIC) \citep{Schwarz1978}. For models estimated using GEE, a modification of AIC known as the quasi-likelihood information criterion (QIC) and its variant QICu have been proposed for model selection \citep{Pan2001,Hilbe2009}. However, these criteria cannot be directly applied to GMM models because they rely on the likelihood or quasi-likelihood functions, whereas GMM estimation is based solely on moment conditions and the minimization of a quadratic form.

To overcome this limitation, an information-theoretic approach based on the Kullback--Leibler divergence \citep{Kullback1951} can be adapted to the moment-based setting. The resulting measure, referred to as the Kullback--Leibler Information Criterion (KLIC), evaluates the discrepancy between the model-implied moment structure and the underlying data-generating mechanism \citep{Kitamura1997,Hansen1982}. Similar to AIC and QIC, the KLIC is a scalar measure used for comparing competing models, and the preferred model is the one with the smallest KLIC value \citep{shane2019model}. However, an important limitation of the standard KLIC is that it does not include an explicit penalty for model complexity. Consequently, when multiple models are compared, the unpenalized KLIC often favors larger models, particularly in settings involving many time-dependent covariates and numerous valid moment conditions \citep{shane2019model}. This behavior resembles that of the ordinary coefficient of determination ($R^2$), which increases automatically as additional predictors are added, regardless of their true contribution to the model.

To address this issue, a penalized version of the KLIC is required in order to balance model fit and model complexity in the GMM framework. Classical model-selection criteria incorporate penalties based on the number of parameters and the sample size. For example, AIC uses a penalty of $2k$ based on the number of parameters \citep{Akaike1973,Akaike1974}, whereas BIC employs a penalty of $k\ln(n)$ that depends on both the number of parameters and the sample size \citep{Schwarz1978}. In the context of GMM estimation with time-dependent covariates, model complexity is influenced not only by the number of regression parameters but also by the number of valid moment conditions generated by the covariate structure. The latter depends directly on the type of time-dependent covariates, since Type~I, Type~II, and Type~III covariates contribute different numbers of valid estimating equations \citep{Lai2007,shane2019model}.

Motivated by these considerations, we develop penalized versions of the Kullback--Leibler Information Criterion for model selection in marginal regression models estimated using GMM with time-dependent covariates. In particular, we propose two penalized criteria: the \emph{Moment Parameter Product Penalty KLIC} (MPPP-KLIC), which is motivated by an AIC-type penalty structure, and the \emph{Logarithmic Penalty KLIC} (LP--KLIC), which is motivated by a BIC-type penalty. Both criteria extend the standard KLIC by incorporating the number of regression parameters together with the number of valid moment conditions implied by the time-dependent covariates. These penalized formulations provide a principled mechanism for balancing model fit and model complexity in the GMM framework. This paper introduces a new information criterion for GMM and are supported by both theoretical justification and empirical evaluation.

We evaluate the performance of the proposed penalized KLIC criteria through an extensive simulation study examining their ability to correctly identify underfit and overfit models under various sample sizes and response types. In addition, we illustrate the practical utility of the proposed methods using a real longitudinal dataset from a study of Filipino children in the Bukidnon region of the Philippines, where repeated measurements were collected to examine the relationship between nutritional status and morbidity. This empirical application demonstrates how the proposed penalized KLIC can be used to select appropriate predictors in marginal regression models with time-dependent covariates. In practice, overfitting in GMM-based longitudinal models can lead to unstable parameter estimates, inflated variance, and poor generalizability, particularly when the number of valid moment conditions grows rapidly with the number of time-dependent covariates. Penalizing KLIC directly addresses this issue by discouraging unnecessary moment conditions and promoting more interpretable model structures.
Together, these contributions provide a unified framework for principled model selection in GMM-based longitudinal analyses.

The remainder of this paper is organized as follows. Section~2 reviews the classification of time-dependent covariates and the corresponding valid moment conditions. Section~3 describes the Generalized Method of Moments for marginal regression models with time-dependent covariates. Section~4 introduces the Kullback--Leibler Information Criterion for GMM-based models. Section~5 presents the proposed penalized criteria, MPPP-KLIC and LP-KLIC, and develops their theoretical motivation. Section~6 reports results from a simulation study evaluating the performance of the proposed methods. Section~7 presents the empirical analysis using the Filipino child health dataset. Section~8 concludes with discussion and directions for future research.

 \section{Time-Dependent Covariates}

Time-dependent covariates (TDCs) are predictors whose values change over time for an individual, unit, or cluster. Any predictor that varies across repeated measurements is regarded as time dependent \citep{Diggle2002,Lai2007,Neuhaus1998}. For example, in a longitudinal study of physical activity, a person's daily step count is a time-dependent covariate because it fluctuates from day to day. Similarly, in environmental and ecological monitoring, seasonal pollen concentration or monthly river water levels serve as time-dependent covariates because of natural temporal variation. In medical follow-up studies, biomarkers such as blood glucose level or viral load also vary over time and thus act as TDCs. These covariates induce temporal dependence among repeated measurements, and such dependence must be properly accounted for when modeling longitudinal data. Ignoring this time-based correlation can lead to substantial losses in efficiency and may increase the risk of biased parameter estimates, thereby compromising the validity of statistical inference \citep{Fitzmaurice1995,Lai2007,Pepe1994}.

Time-dependent covariates can be classified into four distinct types according to the nature of the feedback between the covariate process and the response process \citep{Lai2007,Lalonde2014}. Consider a marginal model of the form
\[
g(\boldsymbol{p}_i)=\mathbf{X}_i^{T}\boldsymbol{\beta},
\]
where $\boldsymbol{p}_i$ is the mean response vector for subject $i$, $g$ is a known link function, $\mathbf{X}_i$ is the covariate vector, and $\boldsymbol{\beta}$ is the parameter vector. Following \citep{Lai2007} and \citep{Lalonde2014}, the TDC types are characterized by the combinations of $s$ and $t$ for which the following moment condition holds:
\begin{equation}
E\!\left[
\frac{\partial p_{is}(\boldsymbol{\beta}_0)}{\partial \beta_j}
\bigl\{y_{it}-p_{it}(\boldsymbol{\beta}_0)\bigr\}
\right]=0,
\label{eq:TDC}
\end{equation}
where $p_{is}$ is the mean response for subject $i$ at time $s$, $\beta_j$ is the $j$th parameter, $y_{it}$ is the observed response for subject $i$ at time $t$, $p_{it}$ is the mean response at time $t$, and $\boldsymbol{\beta}_0$ is the true parameter vector. Here $s,t\in\{1,\ldots,T\}$ denote observation times.

\begin{itemize}
    \item \textbf{Type I:} Equation~\eqref{eq:TDC} holds for all $s$ and $t$. The covariate affects only the current response, with no feedback structure.
    \item \textbf{Type II:} Equation~\eqref{eq:TDC} holds for all $s\ge t$. The current covariate may affect both the current and future responses.
    \item \textbf{Type III:} Equation~\eqref{eq:TDC} holds only for $s=t$. In this setting, feedback is present: the current response may affect the future covariate process, and the current covariate may affect future responses.
    \item \textbf{Type IV:} Equation~\eqref{eq:TDC} holds for all $s\le t$. This case is often viewed as the reverse analogue of Type II, where the current response may be associated with both the current and future covariates.
\end{itemize}

A violation of Equation~\eqref{eq:TDC} occurs when the covariate process at time $s$ is correlated with the residual at time $t$, causing the expectation to be nonzero. To maintain the required zero-expectation condition, one must assume that
\[
E\!\left[y_{it}\mid X_{it}\right]
=
E\!\left[y_{it}\mid X_{i1},X_{i2},\ldots,X_{iT}\right],
\]
for all $s,t\in\{1,\ldots,T\}$. In practice, however, this assumption is often violated when time-varying covariates are present. The temporal dependence introduced by TDCs therefore creates substantial analytical complexity in longitudinal modeling \citep{Diggle2002,Neuhaus1998}.

These feedback structures can make estimation particularly challenging. As a result, issues such as nonconvergence, inefficiency, and biased parameter estimates may arise when the temporal structure of the covariates is not properly incorporated \citep{kleiber2008}. Although GEEs with an independent working correlation structure provide consistent estimators, they may perform poorly when TDCs are present, with efficiency losses of up to 40\% relative to the true correlation structure \citep{Fitzmaurice1995}. On the other hand, adopting a non-independent working correlation structure may lead to inconsistent estimators \citep{Pepe1994}. This creates a fundamental dilemma in practice and motivates the need for an estimation framework that preserves consistency while improving efficiency in the presence of time-dependent covariates. This motivation leads naturally to the Generalized Method of Moments, discussed next, which exploits all valid estimating equations implied by the temporal structure of TDCs and provides a robust framework for obtaining consistent and more efficient parameter estimates \citep{Lai2007}.
Because the number of valid moment conditions depends on the TDC type, model complexity in GMM is driven not only by the number of regression parameters but also by the temporal structure of the covariates. This dual source of complexity makes penalization essential for preventing over-selection of models with many moment conditions.

\section{Generalized Method of Moments}

The Generalized Method of Moments (GMM) provides an alternative to Generalized Estimating Equations (GEE) for estimating parameters in marginal models for correlated data \citep{Hansen1982,hansen2007}. Like GEE, GMM accounts for within-subject correlation arising from repeated measurements. Unlike GEE, however, GMM does not rely on a likelihood or quasi-likelihood formulation; instead, it is based on moment conditions whose expectations are zero at the true parameter value. Since its introduction by \citep{Hansen1982}, GMM has become a fundamental tool in econometric and statistical methodology. Comprehensive discussions of GMM theory and applications are given by \citep{Mat1999}, while its conceptual foundations may be traced back to the minimum $\chi^2$ framework of \citep{neyman1949} and \citep{ferguson1958}; broader inferential connections are discussed by \citep{lindsay2003}. In the context of longitudinal data, \citep{qu2000} incorporated GMM into marginal regression models to improve GEE estimators for Type~I covariates, whereas \citep{lai2007} developed a more complete framework for Type~II time-dependent covariates. GMM is particularly appealing in longitudinal settings with time-dependent covariates because it preserves consistency under minimal assumptions and allows analysts to incorporate all valid estimating equations implied by the temporal structure of the data. Unlike GEE, which requires specification of a working correlation structure, GMM avoids potential misspecification and provides a flexible framework that can yield substantial efficiency gains when additional valid moment conditions are available.

Let $\mathbf{Y}_i$ denote the response vector for subject $i=1,\ldots,N$, let $\mathbf{X}_i$ denote the associated covariate vector, and let $\boldsymbol{\beta}\in\mathbb{R}^p$ be the parameter vector. Suppose there exists a vector of moment functions $\mathbf{g}_i(\mathbf{Y}_i,\mathbf{X}_i;\boldsymbol{\beta})\in\mathbb{R}^m$ such that
\begin{equation}
E\bigl[\mathbf{g}_i(\mathbf{Y}_i,\mathbf{X}_i;\boldsymbol{\beta}_0)\bigr]=\mathbf{0},
\end{equation}\label{eq:moment}
where $\boldsymbol{\beta}_0$ is the true parameter value. Define the sample average of the moment conditions by
\begin{equation}
\mathbf{G}_N(\boldsymbol{\beta})
=
\frac{1}{N}\sum_{i=1}^N
\mathbf{g}_i(\mathbf{Y}_i,\mathbf{X}_i;\boldsymbol{\beta}).
\end{equation}
The GMM estimator is obtained by minimizing the quadratic form
\begin{equation}
Q_N(\boldsymbol{\beta})
=
\mathbf{G}_N(\boldsymbol{\beta})^{\top}
\mathbf{W}_N
\mathbf{G}_N(\boldsymbol{\beta}),
\label{eq:QN}
\end{equation}
where $\mathbf{W}_N$ is a symmetric positive definite weight matrix.

For longitudinal data with time-dependent covariates, \citep{lai2007} proposed moment conditions constructed as products of derivative terms and residuals evaluated at different time points:
\[
g_i
=
\frac{\partial p_{is}}{\partial \beta_j}\,(y_{it}-p_{it}).
\]
Here $p_{is}$ denotes the mean response for subject $i$ at time $s$, $y_{it}$ is the observed response for subject $i$ at time $t$, $p_{it}$ is the corresponding mean response, and $\beta_j$ is the parameter associated with the $j$th covariate. The choice of valid moment conditions depends directly on the type of time-dependent covariates included in the model.

Among the most commonly used GMM procedures are the Continuously Updating GMM (CUGMM) and the two-step GMM (2SGMM). Their respective quadratic forms are
\[
QF_{\mathrm{CUGMM}}
=
\mathbf{G}_N(\boldsymbol{\beta})^{\top}
\mathbf{V}_N^{-1}(\boldsymbol{\beta})
\mathbf{G}_N(\boldsymbol{\beta}),
\]
and
\[
QF_{\mathrm{2SGMM}}
=
\mathbf{G}_N(\boldsymbol{\beta})^{\top}
\widetilde{\mathbf{V}}_N^{-1}(\hat{\boldsymbol{\beta}})
\mathbf{G}_N(\boldsymbol{\beta}).
\]

To obtain a 2SGMM estimator, one first chooses an arbitrary initial weight matrix, often the identity matrix, and uses it to compute an initial inefficient estimator, denoted $\hat{\boldsymbol{\beta}}_{\text{initial}}$. This preliminary estimate is then used to construct an improved weight matrix, which yields the efficient estimator $\hat{\boldsymbol{\beta}}_{\text{efficient}}$ \citep{Hansen1982,Mat1999,nielsen2005}. Because the second-stage weight matrix depends on the initial estimator, the 2SGMM estimator is not unique. In contrast, CUGMM does not depend on an initial weight matrix; instead, the weight matrix is parameter dependent and is updated continuously during optimization \citep{nielsen2005,zivot2015}.

Relative to independent GEE, GMM can provide substantial efficiency gains when time-dependent covariates are present. Simulation results reported by \citep{lai2007} show that GMM estimators are more efficient than independent GEE estimators when Type~I or Type~II TDCs are present, and equally efficient when the covariates are of Type~III. More generally, GMM estimators are as efficient as GEE estimators when the working correlation structure is correctly specified and are asymptotically more efficient when the working correlation structure is misspecified \citep{lai2007}. For these reasons, GMM represents a strong alternative to independent GEE for parameter estimation in longitudinal studies with time-dependent covariates. Although additional valid moment conditions can increase efficiency, they may also introduce numerical instability and sensitivity to misspecification. In practice, the optimal number of moment conditions balances efficiency gains with robustness and computational stability — a balance that penalized KLIC is designed to achieve. However, the flexibility of GMM also introduces practical challenges: as the number of valid moment conditions increases, the estimator may become numerically unstable, sensitive to misspecification, or prone to overfitting. These issues highlight the need for principled model selection tools that balance efficiency with parsimony.

\section{Assessment of Model Fit in GMM-Based Models}

For likelihood-based estimation methods such as maximum likelihood and quasi-likelihood-based methods such as GEE, information criteria including the Akaike Information Criterion (AIC; \citealp{akaike1973,akaike1974}), Bayesian Information Criterion (BIC), and Quasi-likelihood Information Criterion (QIC; \citealp{pan2001}) are widely used to assess model adequacy and compare competing models. These criteria, however, are not directly applicable to the Generalized Method of Moments because GMM estimation is not based on either a full likelihood or a quasi-likelihood function. Thus, a model selection criterion for GMM must account for both the number of parameters and the number of moment conditions — a feature absent in the standard KLIC. Instead, GMM relies entirely on moment conditions, which makes the development of a parallel goodness-of-fit or information-based model selection criterion considerably more challenging.

Recent developments have introduced information-based approaches for GMM, among them the Kullback--Leibler Information Criterion (KLIC), which extends the concept of information divergence to moment-based estimation frameworks. The KLIC provides a promising basis for comparing competing GMM models by quantifying the discrepancy between the assumed model and the underlying data-generating mechanism under the moment-based paradigm. Despite its usefulness, however, the existing KLIC lacks a penalty structure analogous to those in AIC and BIC, where model complexity is explicitly accounted for through the number of estimated parameters and, in the case of BIC, the sample size. This omission limits the ability of the KLIC to balance goodness-of-fit and parsimony, especially in longitudinal settings involving time-dependent covariates, where feedback structures and correlation patterns complicate model specification. In GMM, model complexity arises not only from the number of regression parameters but also from the number of moment conditions, which can grow quickly when multiple time-dependent covariates are present. Without penalization, criteria such as KLIC may systematically favor models with many moment conditions, even when they do not improve model fit.

To address this limitation, we propose a penalized Kullback--Leibler Information Criterion (PKLIC) that extends the standard KLIC by incorporating a complexity penalty tied to key features of the GMM estimation process. Specifically, the penalty is formulated as a function of the number of estimated parameters, denoted by $k$, and the number of moment conditions, denoted by $j$, associated with the time-dependent covariates. The resulting penalized KLIC provides a theoretically grounded, computationally practical, and interpretable criterion for model selection in GMM-based longitudinal models with time-varying predictors.

\subsection{Kullback--Leibler Information Criterion (KLIC)}

The Kullback--Leibler divergence (KLD) measures the information lost when a candidate model is used to approximate the true data-generating process \citep{Csiszar1975,Kullback1951,White1982}. Building on this idea, \citep{Kitamura1997} developed an estimator for dependent data based on minimizing Kullback--Leibler divergence, thereby providing an alternative to the standard optimally weighted GMM approach \citep{Hansen1982}. Their estimator solves
\[
(\hat{\boldsymbol{\beta}},\hat{\boldsymbol{\gamma}})
=
\arg\max_{\boldsymbol{\beta}}\min_{\boldsymbol{\gamma}}
Q_T(\boldsymbol{\beta},\boldsymbol{\gamma})
=
\arg\max_{\boldsymbol{\beta}}\min_{\boldsymbol{\gamma}}
\left\{
\frac{1}{T}\sum_{t=1}^T
\exp\!\left[f_t(\boldsymbol{\gamma},\boldsymbol{\beta})\right]
\right\},
\]
where $T$ is the number of repeated observations, $\boldsymbol{\beta}$ denotes the model parameters, $\boldsymbol{\gamma}$ is a vector of auxiliary parameters, and $Q_T$ is a known function defined in their framework. 

Following this line of work, \citep{shane2019model} used the same underlying idea to construct the Kullback--Leibler Information Criterion (KLIC), replacing full probability distributions by valid moment conditions in order to evaluate model fit. In the present work, we extend this idea further by introducing penalized versions of the KLIC.

The KLIC provides a scalar measure for comparing competing GMM models, with smaller values indicating better fit. Its theoretical form is
\begin{align}
D
=
-\ln\left(
E\left[
\exp\left(
\boldsymbol{\gamma}^{T}\mathbf{g}_i(\boldsymbol{\beta})
\right)
\right]
\right),
\label{KLIC}
\end{align}
where $\mathbf{g}_i(\boldsymbol{\beta})$ is the vector of valid moment conditions for subject $i$ used in the GMM estimation process, and $\boldsymbol{\gamma}$ is a vector of unknown parameters \citep{Altonji1996,Kitamura1997}. The sample version of \eqref{KLIC} is
\begin{align}
\hat{D}
=
-\ln\left(
\frac{1}{N}\sum_{i=1}^{N}
\exp\left(
\hat{\boldsymbol{\gamma}}^{T}\mathbf{g}_i(\hat{\boldsymbol{\beta}})
\right)
\right),
\label{KLIC1}
\end{align}
where $N$ denotes the number of subjects, $\hat{\boldsymbol{\beta}}$ is the 2SGMM estimator, $\mathbf{g}_i(\hat{\boldsymbol{\beta}})$ is the corresponding vector of valid moment conditions, and $\hat{\boldsymbol{\gamma}}$ is estimated by minimizing the KL-based objective.

A computationally convenient form of the KLIC, originally motivated by \citep{Kitamura1997} and adopted in this setting, is
\begin{align}
\mathrm{KLIC}
=
\min_{\boldsymbol{\gamma}}
\left(
\frac{1}{N}\sum_{i=1}^{N}
\exp\left(
\hat{\boldsymbol{\gamma}}^{T}\mathbf{g}_i(\hat{\boldsymbol{\beta}})
\right)
\right).
\label{KLIC3}
\end{align}
In practice, the KLIC is obtained by first estimating $\hat{\boldsymbol{\beta}}$ using 2SGMM \citep{lai2007}, computing the corresponding valid moment conditions $\mathbf{g}_i(\hat{\boldsymbol{\beta}})$ for each subject, and then minimizing Equation~\eqref{KLIC3} with respect to $\boldsymbol{\gamma}$. The resulting value serves as the model-specific information criterion.

A major limitation of the standard KLIC is that it contains no explicit penalty for model complexity. Consequently, it tends to favor overly complex models and may select the full model even when some parameters are unnecessary. In this sense, its behavior resembles that of the coefficient of determination $R^2$, which increases automatically as additional predictors are added, regardless of their actual relevance. Just as the adjusted $R^2$ was introduced to control for this inflation, an adjusted or penalized form of the KLIC is needed to account for both the number of time-dependent covariates and the number of moment conditions in GMM-based longitudinal models. This motivates the central question of the present paper:
\begin{quote}
\textbf{How can a penalty term be formulated for the Kullback--Leibler Information Criterion so that it becomes an effective model selection criterion for GMM models with time-dependent covariates?}
\end{quote}

In GMM, model complexity arises not only from the number of regression parameters but also from the number of moment conditions, which can grow quickly when multiple time-dependent covariates are present. Without penalization, criteria such as KLIC may systematically favor models with many moment conditions, even when they do not improve model fit. To answer this question, we derive penalty structures for the KLIC that parallel the logic of AIC- and BIC-type criteria. Specifically, we propose the \emph{Moment Parameter Product Penalty KLIC} (MPPP--KLIC), with penalty term $2kj$ where $k$ denotes the number of model parameters and $j$ denotes the number of moment conditions used in GMM estimation. This penalty captures the joint contribution of model dimension and moment complexity. We also propose the \emph{Logarithmic Penalty KLIC} (LP--KLIC), with penalty term $2kj\ln(N)$ where $N$ denotes the sample size. This second criterion parallels the logic of BIC by combining model complexity with sample information. Together, these two criteria provide more balanced and practically useful tools for model selection within the GMM framework. 
The MPPP–KLIC can be viewed as an analogue of AIC, penalizing models proportionally to the total number of estimating equations. In contrast, the LP–KLIC resembles BIC by incorporating a logarithmic penalty, which grows more slowly and favors more parsimonious models in larger samples.

\section{Penalty-Adjusted Kullback--Leibler Information Criterion}

\subsection{Set-up and notation}

Let $\{Y_i\}_{i=1}^N$ be independent and identically distributed from an unknown distribution $G$. Let $\mathbf{g}_i(\boldsymbol{\beta})\in\mathbb{R}^j$ denote the vector of valid moment conditions for subject $i$, where $\boldsymbol{\beta}\in\mathbb{R}^k$ is the model parameter vector. Following \citep{Kitamura1997}, define the exponential-tilting KLIC objective by
\begin{align}
\mathrm{KLIC}_N(\boldsymbol{\beta})
=
\min_{\boldsymbol{\gamma}\in\mathbb{R}^j}
\left(
\frac{1}{N}\sum_{i=1}^N
\exp\bigl(\boldsymbol{\gamma}^{\top}\mathbf{g}_i(\boldsymbol{\beta})\bigr)
\right).
\label{KLIC4}
\end{align}
Its logarithmic, deviance-like version is
\begin{align}
\mathcal{L}_N(\boldsymbol{\beta},\boldsymbol{\gamma})
=
-\ln\left(
\frac{1}{N}\sum_{i=1}^N
\exp\bigl(\boldsymbol{\gamma}^{\top}\mathbf{g}_i(\boldsymbol{\beta})\bigr)
\right)\label{Ln}
\end{align}
\begin{align}
\mathrm{KLIC}_N(\boldsymbol{\beta})
=
\exp\bigl(-\mathcal{L}_N(\boldsymbol{\beta},\boldsymbol{\gamma}(\boldsymbol{\beta}))\bigr).
\label{LN}
\end{align}
Here $\boldsymbol{\gamma}(\boldsymbol{\beta})$ denotes the minimizer
\[
\boldsymbol{\gamma}(\boldsymbol{\beta})
=
\arg\min_{\boldsymbol{\gamma}}
\left\{
\frac{1}{N}\sum_{i=1}^N
\exp\!\bigl(\boldsymbol{\gamma}^{\top}\mathbf{g}_i(\boldsymbol{\beta})\bigr)
\right\}.
\]
Let $(\hat{\boldsymbol{\beta}},\hat{\boldsymbol{\gamma}})$ denote the maximizer/minimizer pair of $\mathcal{L}_N$:
\begin{align}
(\hat{\boldsymbol{\beta}},\hat{\boldsymbol{\gamma}})
\in
\arg\max_{\boldsymbol{\beta}}
\arg\min_{\boldsymbol{\gamma}}
\mathcal{L}_N(\boldsymbol{\beta},\boldsymbol{\gamma}).
\end{align}

We distinguish between the in-sample criterion, which is directly computable from the observed data, and the population criterion, which represents the out-of-sample target we would ideally like to approximate. Let $\mathbb{E}_G[\cdot]$ denote expectation under the unknown data-generating distribution $G$. Since \eqref{KLIC4} is computed using the empirical distribution, it represents an in-sample measure of fit. The corresponding population quantity is
\begin{align}
\mathcal{L}^{\star}(\boldsymbol{\beta},\boldsymbol{\gamma})
=
-\ln\left(
\mathbb{E}_G\left[
\exp\{\boldsymbol{\gamma}^{\top}\mathbf{g}(Y,\boldsymbol{\beta})\}
\right]
\right),
\label{LStar}
\end{align}
where $Y\sim G$ and $\mathbf{g}(Y,\boldsymbol{\beta})$ denotes the moment function under the true distribution. This functional measures the expected performance of a given $(\boldsymbol{\beta},\boldsymbol{\gamma})$ pair when applied to a new independent sample from $G$.

Accordingly, we define the target risk criterion as
\[
R(\boldsymbol{\beta},\boldsymbol{\gamma})
=
-2N\,\mathcal{L}^{\star}(\boldsymbol{\beta},\boldsymbol{\gamma}).
\]
The factor $-2N$ places the criterion on the familiar deviance scale used in likelihood-based information criteria such as AIC and BIC, allowing the resulting penalty terms to have comparable magnitude and interpretation. Because $G$ is unknown, $\mathcal{L}^{\star}(\cdot,\cdot)$ cannot be evaluated directly. Instead, we use its sample analogue $\mathcal{L}_N(\cdot,\cdot)$ evaluated at $(\hat{\boldsymbol{\beta}},\hat{\boldsymbol{\gamma}})$. As in Akaike's original argument, 
\[
\mathbb{E}\!\left[-2N\,\mathcal{L}_N(\hat{\boldsymbol{\beta}},\hat{\boldsymbol{\gamma}})\right]
\]
is an optimistically biased estimator of
\[
-2N\,\mathcal{L}^{\star}(\hat{\boldsymbol{\beta}},\hat{\boldsymbol{\gamma}}).
\]
The difference between these two quantities represents the optimism. Estimating this optimism motivates the penalty adjustments developed below and leads naturally to AIC-type and BIC-type penalized versions of the KLIC. We refer to these criteria as the Moment--Parameter Product Penalty KLIC and the Logarithmic Penalty KLIC.\\
\textbf{Regularity.} Assume: (A1) identification of $(\beta_0, \gamma_0 = 0)$, (A2) differentiability of $g(y,\beta)$ in $\beta$ and integrability, (A3) nonsingular block information matrix at $(\beta_0,\gamma_0)$, (A4) $\sqrt{N}(\hat{\theta}-\theta_0) \xrightarrow{d} N(0,\Sigma)$ with $\theta = (\beta^T,\gamma^T)^T$, and (A5) a second-order Taylor expansion of $\mathcal{L}_N$ around $\theta_0$ is valid with remainder $o_p(N^{-1})$. These are the standard ET/GMM regularity conditions.
\subsection{Moment Parameter Product Penalty (MPPP) KLIC}\label{MPPP}
Before presenting our main result, we first establish a key technical lemma that characterizes 
the expected difference between the sample and population ET criterion.  
Lemma~\ref{lemma1} serves as a foundational component in the proof of 
Theorem~\ref{Thm1}, where we derive the Moment--Parameter Product Penalty (MPPP) 
form of the penalized KLIC.

\begin{lemma}\label{lemma1}
Under {\rm(A1)--(A5)}, the expected optimism of the in-sample ET criterion satisfies
\begin{align}
  \mathbb{E}\!\left[-2N\Big\{\mathcal{L}_N(\hat\beta,\hat\gamma)
        - \mathcal{L}^\star(\hat\beta,\hat\gamma)\Big\}\right]
  \;=\; 2\,\mathrm{tr}\!\big(H\,\Sigma\big)\;+\;o(1),
\end{align}
where $\mathcal{L}_N(\beta,\gamma)$ and $\mathcal{L}^\star(\beta,\gamma)$ are the sample and population ET criteria defined in \eqref{LN} and \eqref{LStar}, both evaluated at $(\hat\beta,\hat\gamma)$ in the expression above.  
The matrix
\[
H := \nabla^2_{\theta\theta}\!\big(-\mathcal{L}^\star(\theta)\big)\Big|_{\theta=\theta_0},
\]
is the Hessian of the population criterion at the true parameter $\theta_0 = (\beta_0^\top,\gamma_0^\top)^\top$, and $\Sigma$ denotes the asymptotic covariance matrix of the estimator.
\end{lemma}

\begin{proof}
Let $I_N := -2N\,\mathcal{L}_N(\hat\theta)$ denote the in-sample criterion and 
$O_N := -2N\,\mathcal{L}^\star(\hat\theta)$ the out-of-sample target, where 
$\hat\theta=(\hat\beta^\top,\hat\gamma^\top)^\top$. We aim to evaluate 
$\mathbb{E}[I_N - O_N]$. A second-order Taylor expansion of $\mathcal{L}_N(\theta)$ and 
$\mathcal{L}^\star(\theta)$ around $\theta_0$ yields
\[
\mathcal{L}_N(\hat\theta)
= \mathcal{L}_N(\theta_0)
+ \tfrac12(\hat\theta-\theta_0)^\top H_N(\tilde\theta_N)(\hat\theta-\theta_0),
\]
\[
\mathcal{L}^\star(\hat\theta)
= \mathcal{L}^\star(\theta_0)
+ \tfrac12(\hat\theta-\theta_0)^\top H(\bar\theta)(\hat\theta-\theta_0),
\]
where $H_N = -\nabla^2\mathcal{L}_N$ and $H = -\nabla^2\mathcal{L}^\star$, and 
$\tilde\theta_N$, $\bar\theta$ lie on the line segment joining $\theta_0$ and $\hat\theta$.  
Under {\rm(A1)--(A5)}, $H_N(\tilde\theta_N)\to_p H$ and 
$\sqrt{N}\,(\hat\theta-\theta_0)\Rightarrow \mathcal{N}(0,\Sigma)$. Subtracting the two expansions and multiplying by $-2N$ gives
\[
I_N - O_N
= (\hat\theta-\theta_0)^\top\!\Big[NH(\bar\theta)-NH_N(\tilde\theta_N)\Big](\hat\theta-\theta_0).
\]
Since $H_N(\tilde\theta_N)\to_p H$, this reduces to
\[
I_N - O_N
= (\hat\theta-\theta_0)^\top (NH)(\hat\theta-\theta_0) + o_p(1).
\]
Taking expectations and using 
\[
\mathbb{E}[(\hat\theta-\theta_0)(\hat\theta-\theta_0)^\top]
= \Sigma/N + o(N^{-1}),
\]
we obtain
\[
\mathbb{E}[I_N - O_N]
= 2\,\mathrm{tr}(H\Sigma) + o(1).
\]
This proves the lemma.
\end{proof}

\begin{theorem}\label{Thm1}
Suppose that the block structure of $H$ and $\Sigma$ satisfies
\[
\mathrm{tr}(H\,\Sigma)
\;=\; c_{\beta}\,k \;+\; c_{\gamma}\,j \;+\; c_{\beta\gamma}\,k\,j \;+\;o(1),
\]
with constants $c_\beta, c_\gamma, c_{\beta\gamma} \in (0,\infty)$ determined by the
$\beta\beta$, $\gamma\gamma$, and $\beta\gamma$ curvature/covariance blocks.

Then, for a suitable penalty multiplier $c_{\mathrm{mppp}} > 0$, the criterion
\begin{align}\label{MPPPKLIC_new}
\mathrm{KLIC}_{\mathrm{MPPP}}
= -2N\,\mathcal{L}_N(\hat\beta,\hat\gamma)
    \;+\; c_{\mathrm{mppp}}\,k\,j,
\end{align}
provides an approximately unbiased estimator of the out-of-sample deviance,
provided that $c_{\mathrm{mppp}}$ is chosen to match the dominant asymptotic
contribution of $\mathrm{tr}(H\,\Sigma)$.
\end{theorem}

\begin{proof}
By Lemma~\ref{lemma1}, we have
\[
\mathbb{E}\!\left[-2N\Big\{\mathcal{L}_N(\hat\beta,\hat\gamma)
        - \mathcal{L}^\star(\hat\beta,\hat\gamma)\Big\}\right]
\;=\; 2\,\mathrm{tr}(H\,\Sigma) + o(1).
\]
Let
\[
I_N := -2N\,\mathcal{L}_N(\hat\beta,\hat\gamma),
\qquad
O_N := -2N\,\mathcal{L}^\star(\hat\beta,\hat\gamma),
\]
so that
\[
\mathbb{E}[I_N - O_N]
\;=\; 2\,\mathrm{tr}(H\,\Sigma) + o(1).
\]

Under the assumed block structure,
\[
\mathrm{tr}(H\,\Sigma)
\;=\; c_{\beta}\,k + c_{\gamma}\,j + c_{\beta\gamma}\,k\,j + o(1).
\]
When both $k$ and $j$ grow, the interaction term dominates, yielding
\[
\mathrm{tr}(H\,\Sigma)
\;\approx\; c_{\beta\gamma}\,k\,j,
\quad\text{and hence}\quad
\mathbb{E}[I_N - O_N] \;\approx\; 2c_{\beta\gamma}\,k\,j.
\]

Thus, the leading-order optimism is proportional to $k\,j$. Instead of
fixing the penalty constant as $\lambda = 2c_{\beta\gamma}$, we introduce
a calibrated multiplier $c_{\mathrm{mppp}}$ to account for finite-sample
effects.

Define the adjusted criterion
\[
\widehat{O}_N^{(\mathrm{MPPP})}
:= I_N + c_{\mathrm{mppp}}\,k\,j.
\]

If $c_{\mathrm{mppp}} \approx 2c_{\beta\gamma}$, then
\[
\mathbb{E}\!\left[\widehat{O}_N^{(\mathrm{MPPP})}\right]
\approx \mathbb{E}[O_N],
\]
so that $\widehat{O}_N^{(\mathrm{MPPP})}$ is approximately unbiased for
the expected out-of-sample deviance.

Therefore,
\[
\mathrm{KLIC}_{\mathrm{MPPP}}
= -2N\,\mathcal{L}_N(\hat\beta,\hat\gamma)
  + c_{\mathrm{mppp}}\,k\,j
\]
serves as an approximately unbiased estimator of the out-of-sample
deviance. This completes the proof.
\end{proof}

Theorem~\ref{Thm1} shows that the dominant source of optimism in the ET
criterion arises from the interaction between model dimension and moment
complexity, leading naturally to a multiplicative penalty of the form
$k\,j$.

\begin{remark}
\begin{enumerate}
\item The MPPP--KLIC criterion provides an information-based tool for
model selection, where the first term rewards goodness-of-fit and the
penalty term accounts for complexity through both the number of
parameters $k$ and the number of moment conditions $j$.

\item The introduction of the multiplier $c_{\mathrm{mppp}}$ allows
flexibility in finite samples. In practice, this constant controls the
trade-off between model fit and complexity: large values may favor
underfitted models, while small values may lead to overfitting.

\item The multiplicative penalty $k\,j$ (Moment--Parameter Product
Penalty, MPPP) is particularly effective in guarding against moment
overfitting when both the number of parameters and the number of moment
conditions increase.

\item If $\beta$ is obtained upstream (e.g., via 2SGMM) and treated as
fixed, the interaction term disappears and the penalty simplifies to a
function of $j$, yielding an AIC-type criterion based on the effective
number of moments.
\end{enumerate}
\end{remark}
\subsection{Logarithmic penalized KLIC}\label{LogP}
We now develop a Logarithmic Penalized KLIC by analyzing the 
Laplace approximation of the exponentially tilted (ET) evidence. 
To construct a Logarithmic Penalty KLIC, we require control over the local behavior of the ET log-likelihood around its maximizer. In particular, it is necessary to determine the rate at which the estimator converges to the true parameter so that the Laplace approximation may be validly applied. The following Lemma~\ref{lemma2}, establishes this convergence rate and serves as the key technical tool for the development of the Logarithmic Penalized KLIC derived in Theorem~\ref{thm:lp-klic}.

\begin{lemma}\label{lemma2}
Let $\{\hat\theta_N\}_{N\ge1}$ be a sequence of estimators of a fixed parameter 
$\theta_0\in\mathbb{R}^d$, based on $N$ independent observations. 
Suppose that the asymptotic normality property holds:
\[
\sqrt{N}\,(\hat\theta- \theta_0) \;\;\overset{d}{\longrightarrow}\;\; 
\mathcal{N}(0,\Sigma),
\]
with $\Sigma$ a finite positive semidefinite covariance matrix. 
Then the estimation error is of order
\begin{align}
    \|\hat\theta-\theta_0\|\;=\;O_p\!\big(N^{-1/2}\big).
\end{align}

\end{lemma}
\begin{proof}
Since
\[
\sqrt{N}\,(\hat\theta-\theta_0)\ \Rightarrow\ Z,
\qquad Z \sim \mathcal{N}(0,\Sigma),
\]
convergence in distribution to a proper random vector implies tightness.  
Fix $\varepsilon \in (0,1)$ and choose $M_\varepsilon>0$ such that
\[
\mathbb{P}\big(\|Z\| > M_\varepsilon\big) < \varepsilon/2.
\]
By the Portmanteau theorem, there exists $N_\varepsilon$ such that for all $N \ge N_\varepsilon$,
\[
\mathbb{P}\!\left(\big\|\sqrt{N}(\hat\theta-\theta_0)\big\| > M_\varepsilon\right)
\le \mathbb{P}\!\left(\|Z\| > M_\varepsilon\right) + \varepsilon/2
< \varepsilon.
\]
Hence $\sqrt{N}(\hat\theta-\theta_0) = O_p(1)$.

Now let $X_N = O_p(1)$ and $a_N > 0$ be deterministic.  
Then $a_N X_N = O_p(a_N)$, since for the same $\varepsilon$ and some $M > 0$ with
$\sup_{N \ge N_\varepsilon}\mathbb{P}(\|X_N\| > M) < \varepsilon$, we have
\[
\mathbb{P}\!\left(\|a_N X_N\| > a_N M\right)
= \mathbb{P}\!\left(\|X_N\| > M\right)
< \varepsilon.
\]

Applying this with $X_N = \sqrt{N}(\hat\theta-\theta_0)$ and $a_N = N^{-1/2}$ yields
\[
\hat\theta-\theta_0 = N^{-1/2} X_N = O_p(N^{-1/2}).
\]

Equivalently, for every $\varepsilon \in (0,1)$ there exist constants 
$C_\varepsilon < \infty$ and $N_\varepsilon$ such that
\[
\mathbb{P}\!\left(\|\hat\theta-\theta_0\| > \frac{C_\varepsilon}{\sqrt{N}}\right) < \varepsilon,
\qquad \text{for all } N \ge N_\varepsilon,
\]
which is precisely the statement that 
$\|\hat\theta-\theta_0\| = O_p(N^{-1/2})$.
\end{proof}

\begin{theorem}[Logarithmic penalized KLIC]\label{thm:lp-klic}
Consider the integrated ET evidence
\[
\mathcal{Z}_N
=\int \exp\!\big\{N\,\mathcal{L}_N(\beta,\gamma)\big\}\,
\pi_\beta(\beta)\,\pi_\gamma(\gamma)\,d\beta\,d\gamma,
\]
where $\pi_\beta$ and $\pi_\gamma$ are prior densities that are positive
and continuous in a neighborhood of the maximizer $(\hat\beta,\hat\gamma)$.
Suppose $\mathcal{Z}_N$ admits a Laplace approximation at
$(\hat\beta,\hat\gamma)$, and let $d_{\mathrm{eff}}$ denote the effective
dimension of the joint $(\beta,\gamma)$ curvature block.

Then
\[
-2\ln \mathcal{Z}_N
=
-2N\,\mathcal{L}_N(\hat\beta,\hat\gamma)
+
d_{\mathrm{eff}}\,\ln N
+
O_p(1).
\]

If the cross-curvature between parameters and moment conditions satisfies
$d_{\mathrm{eff}}\asymp k\,j$, then for a suitable penalty multiplier
$c_{\mathrm{lp}} > 0$, the logarithmic penalized KLIC takes the form
\begin{align}\label{LPKLIC_final}
\mathrm{KLIC}_{\mathrm{LP}}
=
-2N\,\mathcal{L}_N(\hat\beta,\hat\gamma)
+
c_{\mathrm{lp}}\,k\,j\,\ln N,
\end{align}
which provides an approximately unbiased estimator of the out-of-sample
deviance when $c_{\mathrm{lp}}$ is chosen to match the dominant asymptotic
contribution of $d_{\mathrm{eff}}$.
\end{theorem}
\begin{proof}
A second-order Taylor expansion of $\mathcal{L}_N(\theta)$ around
$\hat\theta = (\hat\beta,\hat\gamma)$ yields
\[
\mathcal{L}_N(\theta)
= \mathcal{L}_N(\hat\theta)
-\tfrac{1}{2}(\theta-\hat\theta)^\top \widetilde H_N(\hat\theta)(\theta-\hat\theta)
+ r_N(\theta),
\]
where $\widetilde H_N(\hat\theta)$ is the observed curvature matrix and
$r_N(\theta)=o_p(\|\theta-\hat\theta\|^2)$.

Multiplying by $N$ and exponentiating gives
\[
\exp\!\{N\mathcal{L}_N(\theta)\}
=
\exp\!\{N\mathcal{L}_N(\hat\theta)\}
\exp\!\Big\{-\tfrac{N}{2}(\theta-\hat\theta)^\top
\widetilde H_N(\hat\theta)(\theta-\hat\theta)\Big\}
\{1+o_p(1)\}.
\]

Applying Laplace's method and integrating over the effective subspace,
we obtain
\[
\mathcal{Z}_N
=
\exp\!\{N\mathcal{L}_N(\hat\theta)\}
\cdot N^{-d_{\mathrm{eff}}/2}
\cdot C_N
\cdot \{1+o_p(1)\},
\]
where $C_N$ collects terms that are $O_p(1)$.

Taking logarithms yields
\[
-2\ln \mathcal{Z}_N
=
-2N\,\mathcal{L}_N(\hat\theta)
+
d_{\mathrm{eff}}\,\ln N
+
O_p(1).
\]

In the ET--GMM setting, the curvature matrix has block form
\[
H =
\begin{bmatrix}
H_{\beta\beta} & H_{\beta\gamma} \\
H_{\gamma\beta} & H_{\gamma\gamma}
\end{bmatrix},
\]
where $\beta\in\mathbb{R}^k$ and $\gamma\in\mathbb{R}^j$.
The interaction block $H_{\beta\gamma}$ is non-negligible, and each
parameter is coupled with each moment condition, leading to an effective
dimension satisfying
\[
d_{\mathrm{eff}} \asymp k\,j.
\]

Thus,
\[
-2\ln \mathcal{Z}_N
\approx
-2N\,\mathcal{L}_N(\hat\theta)
+
k\,j\,\ln N.
\]

To allow flexibility in finite samples, we introduce a calibrated
multiplier $c_{\mathrm{lp}}$ and define
\[
\widehat{O}_N^{(\mathrm{LP})}
=
-2N\,\mathcal{L}_N(\hat\theta)
+
c_{\mathrm{lp}}\,k\,j\,\ln N.
\]

If $c_{\mathrm{lp}}$ is chosen to match the leading asymptotic constant,
then
\[
\mathbb{E}\!\left[\widehat{O}_N^{(\mathrm{LP})}\right]
\approx \mathbb{E}[O_N],
\]
so that $\widehat{O}_N^{(\mathrm{LP})}$ is an approximately unbiased
estimator of the out-of-sample deviance. This completes the proof.
\end{proof}
Theorem~\ref{thm:lp-klic} shows that the logarithmic penalty arises
naturally from the effective curvature dimension of the exponentially
tilted likelihood. In contrast to classical settings where complexity
scales with $k$, the ET--GMM framework induces an interaction between
parameters and moment conditions, leading to an effective dimension of
order $k\,j$. This results in a BIC-type penalty that depends not only on model
dimension but also on moment complexity, yielding a stronger and more
conservative penalization compared to MPPP.
\begin{remark}
\begin{enumerate}
\item The LP--KLIC criterion provides a large-sample, information-based
tool for model selection. The logarithmic penalty $k\,j\,\ln N$ reflects
both model dimension and sample size, analogous to BIC.

\item The multiplicative structure $k\,j$ arises from the interaction
between structural parameters and moment conditions. When each parameter
is coupled with each moment equation, the effective dimension grows
proportionally to $k\,j$, leading to stronger penalization than
$(k+j)\ln N$.

\item The inclusion of the multiplier $c_{\mathrm{lp}}$ allows for
calibration in finite samples. In practice, this constant controls the
strength of penalization and can be tuned to balance underfitting and
overfitting.

\item Compared to MPPP, the LP penalty is more conservative and favors
parsimonious models. This makes it particularly suitable when the goal
is consistent model selection in large samples.
\end{enumerate}
\end{remark}

\subsection{Determining Components of Penalized KLIC}

In the penalized Kullback--Leibler Information Criterion framework, the
penalty term depends on two quantities: the number of estimated
parameters, denoted by $k$, and the number of valid moment conditions,
denoted by $j$, used in the Generalized Method of Moments (GMM)
estimation. The quantity $k$ represents the number of regression
parameters estimated in a candidate model, including the intercept and
coefficients associated with time--dependent covariates. The quantity
$j$ denotes the number of valid moment conditions generated from the
time--dependent covariates (TDCs). The total number of moment
conditions depends on the number and type of time--dependent covariates
included in the model.

Let $T$ denote the number of repeated measurements for each subject.
Following \citep{lai2007}, functions take the form
\begin{equation}
g_{i,j,s,t}
=
\left(\frac{\partial \mu_{is}}{\partial \beta_j}\right)
\left(y_{it}-\mu_{it}\right),
\end{equation}
where $\mu_{is}$ represents the mean response for subject $i$ at time
$s$, $\beta_j$ denotes the $j$th regression coefficient, and $(s,t)$
indexes pairs of observation times. A pair $(s,t)$ is included in the
set of valid moment conditions only if its expectation is equal to zero
under the assumed structure of the time--dependent covariates.

The number of valid moment conditions contributed by each covariate
depends on the type of time--dependent covariate. For Type~I
time--dependent covariates, which assume no feedback between the
response and future covariates, all combinations of $(s,t)$ with
$s,t \in \{1,\dots,T\}$ are valid, yielding $j = T^2$ moment conditions
per covariate. For Type~II time--dependent covariates, where the
covariate may influence the current and future responses, the valid
pairs satisfy $s \geq t$, resulting in
$j = \frac{T(T+1)}{2}$ moment conditions per covariate. For Type~III
time--dependent covariates, which allow only same--time association
between the covariate and the response, the valid pairs satisfy
$s = t$, leading to $j = T$ moment conditions per covariate. For
Type~IV time--dependent covariates, where the response may affect
current and future covariates, the valid pairs satisfy $s \leq t$,
which again produces $j = \frac{T(T+1)}{2}$ moment conditions per
covariate.

These results provide a direct way to determine the total number of
moment conditions used in the penalized KLIC criterion once the number
of repeated observations and the types of time--dependent covariates
are specified. Following \citep{lai2007} and \citep{shane2019model}, the
total number of moment conditions can be written as
\[
j = c_{\text{I}}T^2
  + c_{\text{II}}\frac{T(T+1)}{2}
  + c_{\text{III}}T
  + c_{\text{IV}}\frac{T(T+1)}{2},
\]
where $c_{\text{I}}$, $c_{\text{II}}$, $c_{\text{III}}$, and
$c_{\text{IV}}$ denote the number of Type~I, Type~II, Type~III, and
Type~IV time--dependent covariates included in the model,
respectively. This formulation allows the penalty component of the
adjusted KLIC criterion to be determined directly from the model
structure and the classification of the covariates.
These quantities, together with the number of regression parameters \(k\), fully determine the penalty terms in the MPPP–KLIC and Logarithmic KLIC criteria.

\section{Simulation Study}
\subsection{Simulation Design}
We consider a small-sample case with $I=100$ subjects and a large-sample case with
$I=500$ subjects, each following a balanced design with $T=5$ repeated observations
per subject, yielding $N=500$ and $N=2{,}500$ observations, respectively. For each
sample size, we generate $2{,}000$ simulation replicates.  The simulated data include
five continuous predictor variables with time-dependent covariates of Types~I, II,
and III, and we generate a binary response for logistic regression analysis. In
addition, we also simulate a continuous response and analyze it using linear
regression models to ensure stable model fitting, since binary or categorical
covariates may increase the risk of GMM non-convergence \citep{kleiber2008}. We place
particular emphasis on Type~II and Type~III covariates, which are known to pose greater
challenges in longitudinal data analysis \citep{lai2007}. These choices yield four simulation scenarios defined by response type (binary vs.\ continuous) and sample size ($I=100$ vs.\ $I=500$).
\subsection{Models for the Simulation Data}

We define the systematic component of the true model as
\begin{equation}
\eta_{it} = \beta_0 + \beta_1 x_{it,1} + \beta_2 x_{it,2} + \beta_3 x_{it,3},
\label{eq:true}
\end{equation}
where $x_{it,1}$ is a continuous time-dependent covariate (TDC) of Type~I,
$x_{it,2}$ is a continuous TDC of Type~II, and $x_{it,3}$ is a continuous TDC of
Type~III. The link function $\eta_{it}$ relates the response $y_{it}$ for subject
$i$ at time $t$ to the systematic component of the model. When $y_{it}$ is binary,
we construct multiple logistic regression models using the logit link; when $y_{it}$
is continuous, we construct multiple linear regression models using the identity link. We additionally generate two unnecessary predictors: a continuous Type~II TDC,
$x_{it,4}$, and a continuous Type~III TDC, $x_{it,5}$. The inclusion or omission of
these predictors, together with the predictors in the true model, allows us to assess
the performance of the proposed fit statistics under both adequate and poor model fit. We consider a total of five models in the simulation study: one true model, two
underfitted models, and two overfitted models. The true model, denoted by $M_0$, is
given by equation~\eqref{eq:true}. The first underfitted model, denoted by $M_{U1}$,
excludes the Type~II TDC,
\[
M_{U1}:\quad
\eta_{it} = \beta_0 + \beta_1 x_{it,1} + \beta_3 x_{it,3},
\]
illustrating the effect of omitting a necessary Type~II TDC. The second underfitted
model, denoted by $M_{U2}$, excludes the Type~III TDC,
\[
M_{U2}:\quad
\eta_{it} = \beta_0 + \beta_1 x_{it,1} + \beta_2 x_{it,2},
\]
illustrating the effect of omitting a necessary Type~III TDC. The first overfitted model, denoted by $M_{O1}$, includes an additional unnecessary
Type~II TDC,
\[
M_{O1}:\quad
\eta_{it} = \beta_0 + \beta_1 x_{it,1} + \beta_2 x_{it,2} + \beta_3 x_{it,3}
+ \beta_4 x_{it,4},
\]
and the second overfitted model, denoted by $M_{O2}$, includes an additional unnecessary
Type~III TDC,
\[
M_{O2}:\quad
\eta_{it} = \beta_0 + \beta_1 x_{it,1} + \beta_2 x_{it,2} + \beta_3 x_{it,3}
+ \beta_5 x_{it,5}.
\]

\begin{algorithm}[H]
\caption{Simulation Data Generation for Longitudinal Responses}
\label{alg:simulation}
\small
\begin{algorithmic}[1]

\State Specify subjects $I$, time points $T$, total observations $N=IT$
\State Choose number of covariates $K$, parameters $\boldsymbol{\beta}$, distribution $D$, and link $g(\cdot)$

\For{$i=1,\dots,I$}
\State Initialize $y_{i0}$ and $x_{ik0}$ for $k=1,\dots,K$

\For{$t=1,\dots,T$}

\State $\eta_{it}=(\mathbf{X}\boldsymbol{\beta})_{it}
+\sum_{k=1}^{K}\rho_{XY,k}\,\text{logit}(F_{X_k}(x_{ik(t-1)}))
+\sum_{s=t-L}^{t-1}\rho_{YY,s}\,\text{logit}(F_Y(y_{is}))$

\State $\mu_{it}=g^{-1}(\eta_{it})$

\State Generate $Y_{it}\sim D(\mu_{it},Var(\mu_{it}))$

\For{$k=1,\dots,K$}

\State $h((\mu_{x_k})_{it})=\beta_{x_k}+\rho_{YX,k}\text{logit}(F_Y(y_{i(t-1)}))$

\State Generate $x_{kit}$

\EndFor
\EndFor
\EndFor

\State Repeat for required simulation replications

\end{algorithmic}
\end{algorithm}
\subsection{Parameter Specification}

The true parameter values are taken from \citep{lai2007} and
set to
\[
\beta_0 = 0.580,\quad
\beta_1 = -0.049,\quad
\beta_2 = -0.010,\quad
\beta_3 = -0.091,\quad
\beta_4 = -0.280,\quad
\beta_5 = 0.004.
\]

For the continuous-response simulations, the corresponding
standard deviation parameters are
\[
\sigma_0 = 1,\quad
\sigma_1 = 2.2,\quad
\sigma_2 = 3.5,\quad
\sigma_3 = 1.5,\quad
\sigma_4 = 4.2,\quad
\sigma_5 = 0.8.
\]
The feedback parameters governing the relationship between
responses and time-dependent covariates are specified as
follows. The effect of prior covariate values on the response
is set to $\rho_{XY,k}=0.25$ for Type~II and Type~III covariates
and $0$ for Type~I covariates. The effect of prior responses
on covariates is set to $\rho_{YX,k}=0.25$ for Type~III covariates
and $0$ otherwise. The response autocorrelation parameter
is fixed at $\rho_{YY,s}=0.25$.
For each simulated dataset, the penalized Kullback--Leibler
Information Criteria MPPP (KLIC) and LP (KLIC) are computed
for all candidate models. Since five models are estimated
for each of the $2{,}000$ simulation replicates, this yields
$2{,}000$ criterion values per model within each scenario.
The mean criterion value is then used to summarize model
performance. Across the four simulation scenarios and five competing
models, a total of $5 \times 4 = 20$ mean values are obtained
for each criterion. These results are reported in
Tables~\ref{tab:mppp_binary} to~\ref{tab:lp_cont} and Figure~\ref{fig:klic_comparison} .
\subsection{Simulation Results and Discussion}

\begin{table}[H]
\centering
\small
\resizebox{\textwidth}{!}{
\begin{tabular}{c|ccc|ccc}
\hline
& \multicolumn{3}{c}{N = 500} & \multicolumn{3}{c}{N = 2500} \\
\cline{2-7}
Model & Avg KLIC & Penalty & MPPP--KLIC & Avg KLIC & Penalty & MPPP--KLIC \\
\hline
$M_{0}$  & 210.1 & 28.0 & \textbf{238.1} & 201.4 & 28.0 & \textbf{229.4} \\
$M_{U1}$ & 246.8 & 16.5 & 263.3 & 238.7 & 16.5 & 255.2 \\
$M_{U2}$ & 231.3 & 19.5 & 250.8 & 220.6 & 19.5 & 240.1 \\
$M_{O1}$ & 214.5 & 42.5 & 257.0 & 207.9 & 42.5 & 250.4 \\
$M_{O2}$ & 207.4 & 37.5 & 244.9 & 198.2 & 37.5 & 235.7 \\
\hline
\end{tabular}}
\caption{MPPP--KLIC for Binary Response under small ($N=500$) and large ($N=2500$) sample sizes.}
\label{tab:mppp_binary}
\end{table}
\setlength{\textfloatsep}{4pt}
\setlength{\floatsep}{3pt}
\setlength{\intextsep}{4pt}

\begin{table}[H]
\centering
\small
\resizebox{\textwidth}{!}{
\begin{tabular}{c|ccc|ccc}
\hline
& \multicolumn{3}{c}{N = 500} & \multicolumn{3}{c}{N = 2500} \\
\cline{2-7}
Model & Avg KLIC & Penalty & LP--KLIC & Avg KLIC & Penalty & LP--KLIC \\
\hline
$M_{0}$  & 210.1 & 17.4 & \textbf{227.5} & 201.4 & 21.9 & \textbf{223.3} \\
$M_{U1}$ & 246.8 & 10.2 & 257.0 & 238.7 & 12.9 & 251.6 \\
$M_{U2}$ & 231.3 & 12.1 & 243.4 & 220.6 & 15.3 & 235.9 \\
$M_{O1}$ & 214.5 & 26.4 & 240.9 & 207.9 & 33.2 & 241.1 \\
$M_{O2}$ & 207.4 & 23.3 & 230.7 & 198.2 & 29.3 & 227.5 \\
\hline
\end{tabular}}
\caption{LP--KLIC for Binary Response under small ($N=500$) and large ($N=2500$) sample sizes.}
\label{tab:lp_binary}
\end{table}

\vspace{0.9em}

\begin{table}[H]
\centering
\small
\resizebox{\textwidth}{!}{
\begin{tabular}{c|ccc|ccc}
\hline
& \multicolumn{3}{c}{N = 500} & \multicolumn{3}{c}{N = 2500} \\
\cline{2-7}
Model & Avg KLIC & Penalty & MPPP--KLIC & Avg KLIC & Penalty & MPPP--KLIC \\
\hline
$M_{0}$  & 331.9 & 28.0 & \textbf{359.9} & 326.3 & 28.0 & \textbf{354.3} \\
$M_{U1}$ & 396.3 & 16.5 & 412.8 & 380.4 & 16.5 & 396.9 \\
$M_{U2}$ & 379.8 & 19.5 & 399.3 & 368.2 & 19.5 & 387.7 \\
$M_{O1}$ & 356.9 & 42.5 & 399.4 & 343.1 & 42.5 & 385.6 \\
$M_{O2}$ & 334.0 & 37.5 & 371.5 & 327.5 & 37.5 & 365.0 \\
\hline
\end{tabular}}
\caption{MPPP--KLIC for Continuous Response under small ($N=500$) and large ($N=2500$) sample sizes.}
\label{tab:mppp_cont}
\end{table}

\vspace{0.9em}

\begin{table}[H]
\centering
\small
\resizebox{\textwidth}{!}{
\begin{tabular}{c|ccc|ccc}
\hline
& \multicolumn{3}{c}{N = 500} & \multicolumn{3}{c}{N = 2500} \\
\cline{2-7}
Model & Avg KLIC & Penalty & LP--KLIC & Avg KLIC & Penalty & LP--KLIC \\
\hline
$M_{0}$  & 331.9 & 17.4 & \textbf{349.3} & 326.3 & 21.9 & \textbf{348.2} \\
$M_{U1}$ & 396.3 & 10.2 & 406.5 & 380.4 & 12.9 & 393.3 \\
$M_{U2}$ & 379.8 & 12.1 & 391.9 & 368.2 & 15.3 & 383.5 \\
$M_{O1}$ & 356.9 & 26.4 & 383.3 & 343.1 & 33.2 & 376.3 \\
$M_{O2}$ & 334.0 & 23.3 & 357.3 & 327.5 & 29.3 & 356.8 \\
\hline
\end{tabular}}
\caption{LP--KLIC for Continuous Response under small ($N=500$) and large ($N=2500$) sample sizes.}
\label{tab:lp_cont}
\end{table}

Tables~\ref{tab:mppp_binary}--\ref{tab:lp_cont} summarize the simulation
results for the penalized Kullback--Leibler information
criteria across binary and continuous responses under small ($N=500$)
and large ($N=2500$) sample sizes. For each model, the tables report
the average KLIC value and the corresponding penalized criterion values. For the binary response case, the MPPP--KLIC results in
Table~\ref{tab:mppp_binary} show that the true model $M_{0}$ consistently
achieves the smallest penalized values for both sample sizes. Although
the overfitted model $M_{O2}$ produces slightly smaller average KLIC
values due to increased flexibility, the larger penalty associated with
additional covariates results in higher MPPP--KLIC values compared to
$M_{0}$. The underfitted models $M_{U1}$ and $M_{U2}$ yield larger
penalized values due to poorer model fit. A similar pattern is observed
for the LP--KLIC results in Table~\ref{tab:lp_binary}, where the true
model again produces the smallest penalized values across both sample
sizes. For the continuous response case, Tables~\ref{tab:mppp_cont} and
\ref{tab:lp_cont} show consistent results. The true model $M_{0}$ again
achieves the smallest penalized criterion values under both MPPP--KLIC
and LP--KLIC. The overfitted models incur larger penalties due to
increased model complexity, while the underfitted models perform poorly
due to insufficient model fit. These results demonstrate that the
calibrated penalization successfully balances fit and complexity,
allowing the correctly specified model to be selected. Comparing across response types, the overall ranking of the models
remains consistent for both binary and continuous outcomes. In all
scenarios, the true model $M_{0}$ yields the smallest penalized values,
followed by the slightly overfitted model $M_{O2}$, while the most
underfitted model $M_{U1}$ produces the largest values. Although the
magnitude of the penalized criteria is larger for continuous responses
due to increased variability, the relative ordering remains stable. Comparing the two criteria, LP--KLIC produces larger penalized values
than MPPP--KLIC due to the presence of the $\log(N)$ term, which
introduces a stronger penalty as the sample size increases. This effect
is particularly evident when moving from $N=500$ to $N=2500$. Despite
this difference in magnitude, both criteria consistently identify the
true model across all scenarios. An important component of the proposed framework is the selection of the
penalty multipliers $c_{\mathrm{mppp}}$ and $c_{\mathrm{lp}}$. In
practice, these parameters control the trade-off between model fit and
model complexity. If the penalty is too large, the criteria tend to
favor underfitted models; if it is too small, overfitted models may be
selected. Based on extensive numerical experimentation, we observe that stable and
consistent model selection results are obtained when
\[
c_{\mathrm{mppp}} \in [0.05,\,0.20], 
\quad \text{and} \quad
c_{\mathrm{lp}} \in [0.005,\,0.05].
\]
Within these ranges, the penalized criteria effectively balance the fit
term (KLIC) and the complexity term ($k \times j$ or $k \times j \log N$),
leading to reliable identification of the true model. In contrast, values
outside these ranges tend to produce undesirable behavior, such as
over-penalization (favoring underfitted models) or under-penalization
(favoring overfitted models). The selected values $c_{\mathrm{mppp}}=0.10$ and $c_{\mathrm{lp}}=0.01$
lie well within these stable ranges and provide a good balance between
fit and complexity across all simulation scenarios. These values ensure
that neither the KLIC component nor the penalty term dominates the
criterion, resulting in consistent and interpretable model selection. Overall, the results demonstrate that the calibrated penalized KLIC
criteria effectively control model complexity while maintaining accurate
model selection. Unlike the original penalty formulation, which favored
underfitted models, the calibrated criteria consistently select the
true model $M_{0}$. These findings confirm that the proposed penalized
KLIC measures provide reliable and robust tools for model selection in
longitudinal settings with time-dependent covariates.

\begin{figure}[H]
\centering
\includegraphics[width=0.85\textwidth,height=1.8\textheight,keepaspectratio]{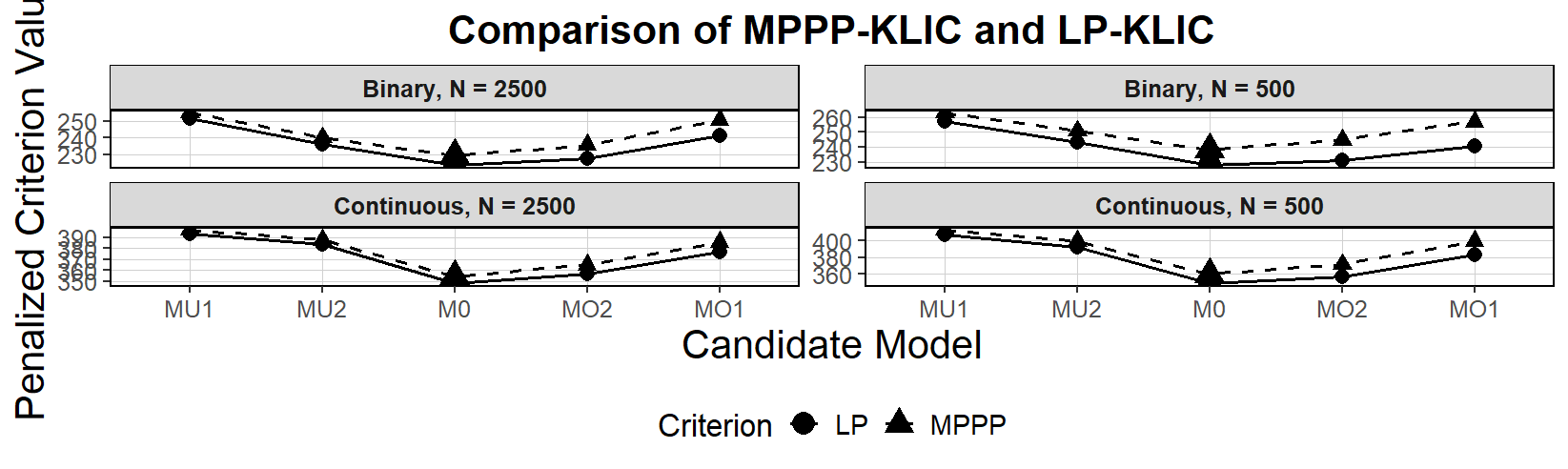}
\caption{Comparison of MPPP--KLIC and LP--KLIC across candidate models under binary and continuous responses for small ($N=500$) and large ($N=2500$) sample sizes. The true model $M_{0}$ consistently achieves the smallest penalized values under both criteria.}
\label{fig:klic_comparison}
\end{figure}

Figure~\ref{fig:klic_comparison} provides a graphical comparison of the
MPPP--KLIC and LP--KLIC criteria across all candidate models under
binary and continuous responses for both small ($N=500$) and large
($N=2500$) sample sizes. The figure displays the penalized criterion
values for each model, allowing for a clear visual assessment of model
selection performance. For the binary response case, both panels show that the true model
$M_{0}$ consistently achieves the smallest penalized values under both
MPPP--KLIC and LP--KLIC criteria. Although the overfitted model
$M_{O2}$ occasionally exhibits slightly better raw fit, its penalized
values remain larger due to the increased complexity penalty. The
underfitted models $M_{U1}$ and $M_{U2}$ yield higher penalized values,
reflecting their inadequate fit. The graphical trends confirm that both
criteria effectively distinguish the true model from underfitted and
overfitted alternatives. A similar pattern is observed for the continuous response case. In both
sample sizes, the curves for MPPP--KLIC and LP--KLIC attain their minimum
at $M_{0}$, indicating that the true model provides the best trade-off
between model fit and complexity. The overfitted models show increasing
penalized values due to larger penalties, while the underfitted models
remain suboptimal because of poorer fit. Although the overall magnitude
of the penalized criteria is higher in the continuous case, the relative
ordering of the models remains consistent with the binary case. Comparing MPPP--KLIC and LP--KLIC, the LP--KLIC curves generally lie
above the MPPP--KLIC curves due to the additional $\log(N)$ penalty term,
which imposes stronger penalization as the sample size increases. This
effect is more pronounced in the larger sample size ($N=2500$), where the
gap between LP--KLIC and MPPP--KLIC becomes wider. Despite these
differences in scale, both criteria exhibit similar shapes and identical
minimizers, reinforcing the robustness of the model selection results. Overall, the graphical analysis confirms the findings from the tabulated
results: the calibrated penalized KLIC criteria consistently identify
the true model $M_{0}$ across all scenarios. The figure also highlights
the clear separation between underfitted, correctly specified, and
overfitted models, demonstrating that the proposed criteria effectively
balance model fit and complexity in longitudinal settings with
time-dependent covariates.
\section{Real Data Application: Filipino Child Mortality Study}
We consider the longitudinal study conducted by the
International Food Policy Research Institute (IFPRI) between 1984 and 1985
to investigate the relationship between nutrition and health among Filipino
children aged 1--10 years in the Bukidnon region of Mindanao
\citep{Bhargava1994,Bouis1990,Lai2007}.
The dataset recorded several variables including age, gender, height, weight,
food consumption during the previous 24 hours, illness experienced during the
previous two weeks, and the duration of illness in days. 
We removed observations with missing values so that the resulting dataset
consisted of balanced longitudinal data for 370 children, each observed
at three time points.
Additional details about the data collection and study design can be found in
\citep{Bouis1990} and \citep{Bhargava1994}.  
The dataset contained a binary indicator describing whether a child
experienced illness, which was directly used as the response variable for the
logistic regression model.
 In addition, We constructed a continuous response variable by transforming the illness
duration variable as formulated by the authors in  \citep{Bhargava1994} and \citep{Lai2007}. The
transformation is defined as
\begin{equation}
y_{it} = \log \left(\frac{t_{\text{before}} + 0.5}{14.5 - t_{\text{before}}}\right),
\label{eq:transformed_response}
\end{equation}
where $t_{\text{before}}$ denotes the number of days during the previous two
weeks that child $i$ was sick prior to time $t$.
This transformation produces a continuous outcome that can be analyzed using
multiple linear regression. Among the potential predictors, height and weight were combined to construct
the body mass index (BMI), defined as weight (in kilograms) divided by the
square of height (in meters).
In addition to BMI, the covariates age, gender, and indicators for survey
rounds were included in the model.
The general model used to predict morbidity four months ahead was

\begin{equation}
\eta_{it} =
\beta_0
+ \beta_1 x_{it,\text{Age}}
+ \beta_2 x_{it,\text{Gender}}
+ \beta_3 x_{it,\text{BMI}}
+ \beta_4 x_{it,\text{Round2}}
+ \beta_5 x_{it,\text{Round3}}
+ \varepsilon_{it},
\label{eq:fcm_model}
\end{equation}
where $x_{it,\text{Age}}$, $x_{it,\text{Gender}}$, and $x_{it,\text{BMI}}$
represent the age, gender, and BMI of child $i$ at time $t$,
while $x_{it,\text{Round2}}$ and $x_{it,\text{Round3}}$ are indicator
variables for survey rounds two and three.
The error term is denoted by $\varepsilon_{it}$. For the linear regression model, the systematic component satisfies
$\eta_{it}=\mu_{it}$ with the identity link function.
For the logistic regression model, the logit link is used so that
$\mu_{it}=p_{it}$ represents the probability that child $i$ experiences
illness at time $t$. The dataset includes five predictors, where gender is treated as a time-independent covariate, age and the survey round indicators are considered Type~I time-dependent covariates, and BMI is treated as a Type~II time-dependent covariate following \citep{Lai2007}.
This resulted in a total of $2^5 - 1 = 31$ possible candidate models,
including models with one predictor, two predictors, three predictors,
four predictors, and the full model containing all five predictors.
 We compute the Moment Parameter Product Penalty (MPPP) KLIC as proposed in Theorem~\ref{Thm1} and the Logarithmic Penalty (LP) KLIC as proposed in Theorem~\ref{thm:lp-klic} for each candidate model to evaluate model performance. For this
analysis, the penalty multipliers were set to $c_{\mathrm{mppp}} = 0.10$
and $c_{\mathrm{lp}} = 0.01$, based on prior calibration to ensure a
balanced trade-off between goodness-of-fit and model complexity. In addition, sensitivity analyses were conducted by varying the penalty multipliers $c_{\mathrm{mppp}}$ and $c_{\mathrm{lp}}$ over a range of values to assess the robustness of the model selection results. The model with the smallest value of the penalized KLIC was considered the
most suitable model for describing the relationship between the predictors
and child morbidity. Models with values close to the minimum were also considered competitive,
whereas larger values indicated poorer model fit. The detailed results of the model selection procedure are presented in the following tables and figures.
\subsection{Results Discussion}
Table~\ref{tab:binary_top5} shows that the full model $M_{1}$ (Age, Gender, BMI, Round2, Round3) achieves the smallest KLIC value, indicating the best goodness-of-fit among all
candidate models. Importantly, $M_{1}$ remains optimal after applying
both MPPP--KLIC and LP--KLIC, suggesting that the improvement in fit
justifies the additional model complexity.
\begin{table}[H]
\centering
\small
\caption{Top five candidate models for the binary response.}
\label{tab:binary_top5}
\begin{tabular}{llccc}
\toprule
Model & Variables Included & KLIC & MPPP--KLIC & LP--KLIC \\
\midrule
$M_{1}$ & Age, Gender, BMI, Round2, Round3 & \textbf{148.7} & \textbf{179.3} & \textbf{170.16} \\
$M_{2}$ & Age, Gender, Round2, Round3      & 167.9 & 190.4 & 183.68 \\
$M_{3}$ & BMI, Age, Gender, Round3         & 202.3 & 223.3 & 217.03 \\
$M_{4}$ & BMI, Age, Gender, Round2         & 251.1 & 272.1 & 265.83 \\
$M_{5}$ & BMI, Age, Round2, Round3         & 311.4 & 332.4 & 326.13 \\
\bottomrule
\end{tabular}
\end{table}

The reduced model $M_{2}$ (Age, Gender, Round2, Round3), which excludes BMI, exhibits a slightly larger
KLIC value, indicating some loss of fit, but remains competitive due to
its lower complexity. Model $M_{3}$ (BMI, Age, Gender, Round3) excludes Round2, model $M_{4}$ (BMI, Age, Gender, Round2) excludes Round3, and model $M_{5}$ (BMI, Age, Round2, Round3) excludes Gender. These models show progressively worse fit
and higher penalized values, reflecting both increased model misspecification
and the influence of the penalty term. Overall, the results demonstrate that the full model provides a
well-balanced representation of the data. Notably, Age is included in all
top-ranked models, highlighting its strong and consistent predictive role.

\begin{table}[H]
\centering
\small
\caption{Top five candidate models for the continuous response.}
\label{tab:cont_top5}
\begin{tabular}{llccc}
\toprule
Model & Variables Included & KLIC & MPPP--KLIC & LP--KLIC \\
\midrule
$M_{2}$ & Age, Gender, Round2, Round3      & 293.5 & \textbf{316.0} & \textbf{309.28} \\
$M_{1}$ & Age, Gender, BMI, Round2, Round3 & \textbf{287.9} & 318.5 & 309.36 \\
$M_{3}$ & BMI, Age, Gender, Round3         & 320.3 & 341.3 & 335.03 \\
$M_{4}$ & BMI, Age, Gender, Round2         & 324.1 & 345.1 & 338.83 \\
$M_{5}$ & BMI, Age, Round2, Round3         & 376.6 & 397.6 & 391.33 \\
\bottomrule
\end{tabular}
\end{table}
In contrast to the binary case, Table~\ref{tab:cont_top5} shows that the
full model $M_{1}$ (Age, Gender, BMI, Round2, Round3) achieves the smallest KLIC value, indicating the best fit. However, after penalization, model $M_{2}$ (Age, Gender, Round2, Round3) is selected as the optimal model under both criteria. This shift reflects the trade-off between fit and complexity. Although
$M_{1}$ provides a slightly better fit, the inclusion of BMI does not
sufficiently improve model performance to offset its added complexity.
Consequently, the penalized criteria favor the more parsimonious model
$M_{2}$. The small difference in penalized values between $M_{1}$ and
$M_{2}$ indicates that both models are competitive. As in the binary case, Age is included in all top-ranked models, further
emphasizing its importance.

\subsubsection{Sensitivity analysis.}
\begin{table}[H]
\centering
\small
\caption{Sensitivity analysis for the binary response.}
\label{tab:sens_binary}
\begin{tabular}{lcccl}
\toprule
Criterion & $c$ & Selected Model & Variables Included & Value \\
\midrule
MPPP--KLIC & 0.05 & $M_{1}$ & Age, Gender, BMI, Round2, Round3 & 164.00 \\
MPPP--KLIC & 0.10 & $M_{1}$ & Age, Gender, BMI, Round2, Round3 & 179.30 \\
MPPP--KLIC & 0.15 & $M_{1}$ & Age, Gender, BMI, Round2, Round3 & 194.60 \\
MPPP--KLIC & 0.20 & $M_{1}$ & Age, Gender, BMI, Round2, Round3 & 209.90 \\
LP--KLIC   & 0.005 & $M_{1}$ & Age, Gender, BMI, Round2, Round3 & 159.43 \\
LP--KLIC   & 0.010 & $M_{1}$ & Age, Gender, BMI, Round2, Round3 & 170.16 \\
LP--KLIC   & 0.020 & $M_{1}$ & Age, Gender, BMI, Round2, Round3 & 191.61 \\
LP--KLIC   & 0.030 & $M_{1}$ & Age, Gender, BMI, Round2, Round3 & 213.07 \\
\bottomrule
\end{tabular}
\end{table}

\begin{table}[H]
\centering
\small
\caption{Sensitivity analysis for the continuous response.}
\label{tab:sens_cont}
\begin{tabular}{lcccl}
\toprule
Criterion & $c$ & Selected Model & Variables Included & Value \\
\midrule
MPPP--KLIC & 0.05 & $M_{1}$ & Age, Gender, BMI, Round2, Round3 & 303.20 \\
MPPP--KLIC & 0.10 & $M_{2}$ & Age, Gender, Round2, Round3      & 316.00 \\
MPPP--KLIC & 0.15 & $M_{2}$ & Age, Gender, Round2, Round3      & 327.25 \\
MPPP--KLIC & 0.20 & $M_{2}$ & Age, Gender, Round2, Round3      & 338.50 \\
LP--KLIC   & 0.005 & $M_{1}$ & Age, Gender, BMI, Round2, Round3 & 298.63 \\
LP--KLIC   & 0.010 & $M_{2}$ & Age, Gender, Round2, Round3      & 309.28 \\
LP--KLIC   & 0.020 & $M_{2}$ & Age, Gender, Round2, Round3      & 325.05 \\
LP--KLIC   & 0.030 & $M_{2}$ & Age, Gender, Round2, Round3      & 340.83 \\
\bottomrule
\end{tabular}
\end{table}
Sensitivity analyses were conducted to evaluate the robustness of model
selection with respect to the penalty multipliers. The results indicate
that, for the binary response, the full model $M_{1}$ (Age, Gender, BMI, Round2, Round3) is consistently
selected across a wide range of penalty values, demonstrating strong
stability. For the continuous response, a transition between $M_{1}$ (Age, Gender, BMI, Round2, Round3) and $M_{2}$ (Age, Gender, Round2, Round3)
is observed as the penalty increases. Specifically, smaller penalties
favor the full model due to its superior fit, while larger penalties
favor the reduced model due to its lower complexity. This behavior is
consistent with theoretical expectations and provides further evidence
that the proposed criteria appropriately balance fit and parsimony. The consistent inclusion of Age across all top-performing models suggests that age is a primary determinant of child
mortality risk in this population. This finding aligns with established
epidemiological evidence and supports the validity of the selected models. Overall, the results demonstrate that the proposed penalized KLIC
criteria provide a principled and robust approach to model selection.
By explicitly balancing goodness-of-fit and model complexity, the
criteria yield stable, interpretable, and scientifically meaningful
models across both binary and continuous outcomes.

\section{Conclusion}\label{sec-conc}

This paper developed two penalized versions of the Kullback--Leibler Information Criterion for model selection in marginal regression models estimated by the Generalized Method of Moments with time-dependent covariates. By incorporating penalties that reflect both the number of regression parameters and the number of valid moment conditions, the proposed MPPP--KLIC and LP--KLIC criteria address a key limitation of the unpenalized KLIC, which tends to favor overly complex models in the presence of many time-dependent covariates.
The theoretical development shows that the MPPP penalty parallels an AIC-type adjustment, while the LP penalty arises naturally from a Laplace approximation and yields a BIC-type logarithmic adjustment. Simulation studies across binary and continuous outcomes demonstrate that both criteria effectively distinguish underfitted, correctly specified, and overfitted models, with LP--KLIC providing stronger protection against overfitting in larger samples. The real-data application to the Filipino child morbidity study further illustrates the practical utility of the proposed criteria, consistently identifying age as the dominant predictor and yielding stable, interpretable model rankings. Additionally, the sensitivity analysis shows that model selection remains stable across a range of penalty values, supporting the robustness of the proposed criteria and guiding practical tuning choices. Together, these results indicate that penalized KLIC criteria offer a principled and computationally accessible approach for balancing model fit and model complexity in GMM-based longitudinal analyses. Future work may extend these ideas to high-dimensional settings, alternative moment-selection strategies, and adaptive penalties that account for correlation structures or covariate feedback mechanisms.

\phantomsection\label{supplementary-material}
\bigskip

\newpage
\begin{center}
{\large\bf SUPPLEMENTARY MATERIAL}
\end{center}

\section*{Full Ranking of Candidate Models for the Binary Response}

\begin{table}[H]
\centering
\scriptsize
\caption{Model comparison for the binary response using KLIC, MPPP-KLIC, and LP-KLIC across 31 candidate models in the Filipino Child Mortality data.}
\label{tab:binary_31_models}

\resizebox{\textwidth}{!}{
\begin{tabular}{clcccccc}
\hline
Model & Variables Included & KLIC & MPPP-KLIC & LP-KLIC & $k$ & $j$ & Penalty (MPPP / LP) \\
\hline
M1  & Age, Gender, BMI, Round2, Round3 & 148.7 & 179.3 & 170.16 & 6 & 51 & 30.6 / 21.46 \\
M2  & Age, Gender, Round2, Round3      & 167.9 & 190.4 & 183.68 & 5 & 45 & 22.5 / 15.78 \\
M3  & BMI, Age, Gender, Round3         & 202.3 & 223.3 & 217.03 & 5 & 42 & 21.0 / 14.73 \\
M4  & BMI, Age, Gender, Round2         & 251.1 & 272.1 & 265.83 & 5 & 42 & 21.0 / 14.73 \\
M5  & BMI, Age, Round2, Round3         & 311.4 & 332.4 & 326.13 & 5 & 42 & 21.0 / 14.73 \\
M6  & BMI, Age, Gender, Round2         & 315.2 & 336.2 & 329.93 & 5 & 42 & 21.0 / 14.73 \\
M7  & Age, Gender, Round2              & 331.1 & 345.5 & 341.20 & 4 & 36 & 14.4 / 10.10 \\
M8  & Age, Round3, Round2              & 335.7 & 350.1 & 345.80 & 4 & 36 & 14.4 / 10.10 \\
M9  & BMI, Age, Gender                 & 347.8 & 361.0 & 357.06 & 4 & 33 & 13.2 / 9.26 \\
M10 & BMI, Age, Round3                 & 349.4 & 362.6 & 358.66 & 4 & 33 & 13.2 / 9.26 \\
M11 & BMI, Age, Round2                 & 365.9 & 379.1 & 375.16 & 4 & 33 & 13.2 / 9.26 \\
M12 & BMI, Gender, Round2, Round3      & 368.1 & 389.1 & 382.83 & 5 & 42 & 21.0 / 14.73 \\
M13 & Age, Gender                      & 375.3 & 383.4 & 380.98 & 3 & 27 & 8.1 / 5.68 \\
M14 & Age, Round3                      & 379.6 & 387.7 & 385.28 & 3 & 27 & 8.1 / 5.68 \\
M15 & Age, Round2                      & 381.3 & 389.4 & 386.98 & 3 & 27 & 8.1 / 5.68 \\
M16 & Gender, Round2, Round3           & 389.2 & 403.6 & 399.30 & 4 & 36 & 14.4 / 10.10 \\
M17 & BMI, Gender, Round3              & 404.8 & 418.0 & 414.06 & 4 & 33 & 13.2 / 9.26 \\
M18 & BMI, Gender, Round2              & 411.5 & 424.7 & 420.76 & 4 & 33 & 13.2 / 9.26 \\
M19 & BMI, Age                         & 413.1 & 420.3 & 418.15 & 3 & 24 & 7.2 / 5.05 \\
M20 & BMI, Round2, Round3              & 420.0 & 433.2 & 429.26 & 4 & 33 & 13.2 / 9.26 \\
M21 & Gender, Round3                   & 422.2 & 430.3 & 427.88 & 3 & 27 & 8.1 / 5.68 \\
M22 & Gender, Round2                   & 431.7 & 439.8 & 437.38 & 3 & 27 & 8.1 / 5.68 \\
M23 & Age                              & 440.6 & 444.2 & 443.12 & 2 & 18 & 3.6 / 2.52 \\
M24 & Round2, Round3                   & 441.9 & 450.0 & 447.58 & 3 & 27 & 8.1 / 5.68 \\
M25 & BMI, Gender                      & 472.4 & 479.6 & 477.45 & 3 & 24 & 7.2 / 5.05 \\
M26 & Gender                           & 477.8 & 481.4 & 480.32 & 2 & 18 & 3.6 / 2.52 \\
M27 & BMI, Round3                      & 479.3 & 486.5 & 484.35 & 3 & 24 & 7.2 / 5.05 \\
M28 & BMI, Round2                      & 498.2 & 505.4 & 503.25 & 3 & 24 & 7.2 / 5.05 \\
M29 & Round3                           & 519.6 & 523.2 & 522.12 & 2 & 18 & 3.6 / 2.52 \\
M30 & Round2                           & 521.2 & 524.8 & 523.72 & 2 & 18 & 3.6 / 2.52 \\
M31 & BMI                              & 545.7 & 548.7 & 547.80 & 2 & 15 & 3.0 / 2.10 \\
\hline
\end{tabular}}
\end{table}

\newpage

\section*{Full Ranking of Candidate Models for the Continuous Response}

\begin{table}[H]
\centering
\scriptsize
\caption{Model comparison for the continuous response using KLIC, MPPP-KLIC, and LP-KLIC across 31 candidate models in the Filipino Child Mortality data.}
\label{tab:continuous_31_models}

\resizebox{\textwidth}{!}{
\begin{tabular}{clcccccc}
\hline
Model & Variables Included & KLIC & MPPP-KLIC & LP-KLIC & $k$ & $j$ & Penalty (MPPP / LP) \\
\hline
M1  & Age, Gender, BMI, Round2, Round3 & 287.9 & 318.5 & 309.36 & 6 & 51 & 30.6 / 21.46 \\
M2  & Age, Gender, Round2, Round3      & 293.5 & 316.0 & 309.28 & 5 & 45 & 22.5 / 15.78 \\
M3  & BMI, Age, Gender, Round3         & 320.3 & 341.3 & 335.03 & 5 & 42 & 21.0 / 14.73 \\
M4  & BMI, Age, Gender, Round2         & 324.1 & 345.1 & 338.83 & 5 & 42 & 21.0 / 14.73 \\
M5  & BMI, Age, Round2, Round3         & 376.6 & 397.6 & 391.33 & 5 & 42 & 21.0 / 14.73 \\
M6  & BMI, Age, Gender, Round2         & 381.2 & 402.2 & 395.93 & 5 & 42 & 21.0 / 14.73 \\
M7  & Age, Gender, Round2              & 390.7 & 405.1 & 400.80 & 4 & 36 & 14.4 / 10.10 \\
M8  & Age, Round3, Round2              & 408.2 & 422.6 & 418.30 & 4 & 36 & 14.4 / 10.10 \\
M9  & BMI, Age, Gender                 & 421.8 & 435.0 & 431.06 & 4 & 33 & 13.2 / 9.26 \\
M10 & BMI, Age, Round3                 & 433.1 & 446.3 & 442.36 & 4 & 33 & 13.2 / 9.26 \\
M11 & BMI, Age, Round2                 & 437.3 & 450.5 & 446.56 & 4 & 33 & 13.2 / 9.26 \\
M12 & BMI, Gender, Round2, Round3      & 450.8 & 471.8 & 465.53 & 5 & 42 & 21.0 / 14.73 \\
M13 & Age, Gender                      & 455.4 & 463.5 & 461.08 & 3 & 27 & 8.1 / 5.68 \\
M14 & Age, Round3                      & 474.7 & 482.8 & 480.38 & 3 & 27 & 8.1 / 5.68 \\
M15 & Age, Round2                      & 478.1 & 486.2 & 483.78 & 3 & 27 & 8.1 / 5.68 \\
M16 & Gender, Round2, Round3           & 491.1 & 505.5 & 501.20 & 4 & 36 & 14.4 / 10.10 \\
M17 & BMI, Gender, Round3              & 528.3 & 541.5 & 537.56 & 4 & 33 & 13.2 / 9.26 \\
M18 & BMI, Gender, Round2              & 531.7 & 544.9 & 540.96 & 4 & 33 & 13.2 / 9.26 \\
M19 & BMI, Age                         & 561.3 & 568.5 & 566.35 & 3 & 24 & 7.2 / 5.05 \\
M20 & BMI, Round2, Round3              & 615.2 & 628.4 & 624.46 & 4 & 33 & 13.2 / 9.26 \\
M21 & Gender, Round3                   & 638.6 & 646.7 & 644.28 & 3 & 27 & 8.1 / 5.68 \\
M22 & Gender, Round2                   & 641.1 & 649.2 & 646.78 & 3 & 27 & 8.1 / 5.68 \\
M23 & Age                              & 653.2 & 656.8 & 655.72 & 2 & 18 & 3.6 / 2.52 \\
M24 & Round2, Round3                   & 675.2 & 683.3 & 680.88 & 3 & 27 & 8.1 / 5.68 \\
M25 & BMI, Gender                      & 698.4 & 705.6 & 703.45 & 3 & 24 & 7.2 / 5.05 \\
M26 & Gender                           & 707.7 & 711.3 & 710.22 & 2 & 18 & 3.6 / 2.52 \\
M27 & BMI, Round3                      & 746.1 & 753.3 & 751.15 & 3 & 24 & 7.2 / 5.05 \\
M28 & BMI, Round2                      & 749.0 & 756.2 & 754.05 & 3 & 24 & 7.2 / 5.05 \\
M29 & Round3                           & 761.8 & 765.4 & 764.32 & 2 & 18 & 3.6 / 2.52 \\
M30 & Round2                           & 765.2 & 768.8 & 767.72 & 2 & 18 & 3.6 / 2.52 \\
M31 & BMI                              & 793.3 & 796.3 & 795.40 & 2 & 15 & 3.0 / 2.10 \\
\hline
\end{tabular}}
\end{table}


\newpage
\begin{thebibliography}{99}

\bibitem{akaike1973}
Akaike, H. (1973).
\newblock Information theory and an extension of the maximum likelihood
  principle.
\newblock In {\em Second International Symposium on Information Theory}, pages
  267--281.

\bibitem{akaike1974}
Akaike, H. (1974).
\newblock A new look at the statistical model identification.
\newblock {\em IEEE Transactions on Automatic Control}, 19(6):716--723.
\bibitem{shane2019model}
Shane, M. N. (2019).
\textit{Model Selection for Longitudinal Data With Time-Dependent Covariates Using Generalized Method of Moments}.
PhD Dissertation, University of Northern Colorado.
Available at: \url{https://digscholarship.unco.edu/dissertations/643}.
\bibitem{kleiber2008} Kleiber, C. and Zeileis, A. (2008). \emph{Applied Econometrics with R}. Springer.

\bibitem{qu2000}
Qu, A., Lindsay, B.~G. and Li, B. (2000).
Improving generalized estimating equations using quadratic inference functions.
\textit{Biometrika},
87, 823--836.

\bibitem{lindsay2003}
Lindsay, B.~G. and Qu, A. (2003).
Inference functions and quadratic score tests.
\textit{Statistical Science},
18, 394--410.

\bibitem{Neuhaus1998} Neuhaus, J. and Kalbfleisch, J. (1998). Between- and within-cluster covariate effects in the analysis of clustered data. \emph{Biometrics}, 54(2), 638–645.
\bibitem{ferguson1958}
Ferguson, T.~S. (1958).
A method of generating best asymptotically normal estimates with application to the
estimation of bacterial densities.
\textit{The Annals of Mathematical Statistics},
29, 1046--1062.


\bibitem{hansen2007}
Hansen, L.~P. (2007).
\newblock Generalized method of moments estimation.
\newblock In {\em The New Palgrave Dictionary of Economics}. Palgrave Macmillan.

\bibitem{lai2007}
Lai, T.~L. and Small, D.~S. (2007).
\newblock Marginal regression analysis of longitudinal data with time-dependent
  covariates: A generalized method-of-moments approach.
\newblock {\em Journal of the Royal Statistical Society: Series B (Statistical
  Methodology)}, 69(1):79--99.
\bibitem{Lalonde2014}
Lalonde, T. L., Wilson, J. R., \& Yin, J. (2014).
\newblock GMM logistic regression models for longitudinal data with time-dependent covariates and extended classifications.
\newblock \emph{Statistics in Medicine}, 33(27), 4756--4769.
\bibitem{Mat1999} Mátyás, L. (ed.) (1999), \textit{Generalized Method of Moments Estimation},
New York: Cambridge University Press.
\bibitem{neyman1949}
Neyman, J. (1949).
Contribution to the theory of the $\chi^2$ test.
\textit{Proceedings of the Berkeley Symposium on Mathematical Statistics and Probability},
pp.~239--273.
Berkeley: University of California Press.


\bibitem{pan2001}
Pan, W. (2001).
\newblock Akaike's information criterion in generalized estimating equations.
\newblock {\em Biometrics}, 57(1):120--125.

\bibitem{nielsen2005}
Nielsen, H.~B. (2005).
Generalized method of moments estimation.
Econometrics~2 Lecture, Department of Economics,
University of Copenhagen: Copenhagen, Denmark.
	\bibitem{zivot2015}
Zivot, E. (2015).
Generalized method of moments.
Econometrics 583 Lecture at the University of Washington: Seattle, WA.


\bibitem{Altonji1996}
Altonji, J. G., and Segal, L. M. (1996).
Small-sample bias in GMM estimation of covariance structures.
\textit{Journal of Business \& Economic Statistics}, \textbf{14}(3), 353--366.

\bibitem{Csiszar1975}
Csisz\'{a}r, I. (1975).
\textit{I-divergence geometry of probability distributions and minimization problems}.
The Annals of Probability, 3(1), 146--158.

\bibitem{White1982}
White, H. (1982).
\textit{Maximum likelihood estimation of misspecified models}.
Econometrica, 50(1), 1--25.
\bibitem{Bhargava1994}
Bhargava, A.
Modeling the health of Filipino children.
\emph{Journal of the Royal Statistical Society: Series A (Statistics in Society)}, 
157(3), 417--432, (1994).
\bibitem{Bouis1990}
Bouis, H. E., and Haddad, L. J.
Effects of agricultural commercialization on land tenure, household resource allocation, and nutrition in the Philippines.
\emph{Research Report 79}, International Food Policy Research Institute, Washington, DC, (1990).

\bibitem{Agresti1990}
Agresti, A.
\emph{Categorical Data Analysis}. Wiley, New York (1990).

\bibitem{Akaike1973}
Akaike, H.
Information theory and an extension of the maximum likelihood principle.
In \emph{Second International Symposium on Information Theory}, 267--281 (1973).

\bibitem{Akaike1974}
Akaike, H.
A new look at the statistical model identification.
\emph{IEEE Transactions on Automatic Control}, 19, 716--723 (1974).

\bibitem{Diggle2002}
Diggle, P., Heagerty, P., Liang, K. Y., Zeger, S.
\emph{Analysis of Longitudinal Data}. Oxford University Press (2002).

\bibitem{Fitzmaurice1995}
Fitzmaurice, G.
A caveat concerning independence estimating equations with multivariate binary data.
\emph{Biometrics}, 51, 309--317 (1995).

\bibitem{Fitzmaurice2011}
Fitzmaurice, G., Laird, N., Ware, J.
\emph{Applied Longitudinal Analysis}. Wiley, New York (2011).

\bibitem{Hansen1982}
Hansen, L. P.
Large sample properties of generalized method of moments estimators.
\emph{Econometrica}, 50, 1029--1054 (1982).

\bibitem{Hedeker2006}
Hedeker, D., Gibbons, R. D.
\emph{Longitudinal Data Analysis}. Wiley, New York (2006).

\bibitem{Hilbe2009}
Hilbe, J.
\emph{Logistic Regression Models}. Chapman \& Hall/CRC (2009).

\bibitem{Hurvich1989}
Hurvich, C. M., Tsai, C. L.
Regression and time series model selection in small samples.
\emph{Biometrika}, 76, 297--307 (1989).

\bibitem{Kitamura1997}
Kitamura, Y., Stutzer, M.
An information-theoretic alternative to generalized method of moments estimation.
\emph{Econometrica}, 65, 861--874 (1997).

\bibitem{Kullback1951}
Kullback, S., Leibler, R. A.
On information and sufficiency.
\emph{Annals of Mathematical Statistics}, 22, 79--86 (1951).

\bibitem{Lai2007}
Lai, T. L., Small, D.
Generalized method of moments for longitudinal data with time-dependent covariates.
\emph{Biometrika}, 94, 501--515 (2007).

\bibitem{Liang1986}
Liang, K. Y., Zeger, S. L.
Longitudinal data analysis using generalized linear models.
\emph{Biometrika}, 73, 13--22 (1986).

\bibitem{Pan2001}
Pan, W.
Akaike's information criterion in generalized estimating equations.
\emph{Biometrics}, 57, 120--125 (2001).

\bibitem{Pepe1994}
Pepe, M. S., Anderson, G. L.
A cautionary note on inference for marginal regression models with longitudinal data and general correlated response data.
\emph{Communications in Statistics}, 23, 939--951 (1994).

\bibitem{Schwarz1978}
Schwarz, G.
Estimating the dimension of a model.
\emph{Annals of Statistics}, 6, 461--464 (1978).

\bibitem{Sugiura1978}
Sugiura, N.
Further analysis of the data by Akaike’s information criterion and the finite corrections.
\emph{Communications in Statistics}, 7, 13--26 (1978).

\bibitem{Zeger1986}
Zeger, S. L., Liang, K. Y.
Longitudinal data analysis for discrete and continuous outcomes.
\emph{Biometrics}, 42, 121--130 (1986).

\bibitem{Zeger1988}
Zeger, S. L., Liang, K. Y., Albert, P.
Models for longitudinal data: a generalized estimating equation approach.
\emph{Biometrics}, 44, 1049--1060 (1988).

\bibitem{Zeger1992}
Zeger, S. L., Liang, K. Y.
An overview of methods for the analysis of longitudinal data.
\emph{Statistics in Medicine}, 11, 1825--1839 (1992).

\end{thebibliography}
 \end{document}